\documentclass{llncs}

\usepackage{thm-restate}

\usepackage[utf8]{inputenc}
\usepackage[T1]{fontenc}
\usepackage{underscore}
\usepackage{tikz}
\usepackage{amsmath}
\usepackage{stmaryrd}
\usepackage{amssymb}
\usepackage{amscd}
\usepackage{upref}
\usepackage[all]{xy}
\usepackage{hyperref}
\usepackage{algorithm}
\usepackage{algorithmic}
\usepackage{xifthen}
\usepackage{xparse}
\usepackage{forest}
\usepackage{xspace}
\usepackage{relsize}
\usepackage{wrapfig}
\usepackage{graphicx}
\graphicspath{{figures/}}
\usepackage{xparse} % for NewDocumentCommand in macros
\usepackage{macros}
\usepackage{cleveref}
\usepackage{threeparttable}
\usepackage{caption}
\usepackage{subcaption}

\usepackage{paralist}
\usepackage[shortlabels]{enumitem}

\usepackage{slashbox} % for table
\usepackage{multicol} % for algorithms

\usepackage{upgreek}
\usetikzlibrary{arrows,shapes,automata,backgrounds,petri}
\usetikzlibrary{automata,positioning}
\usetikzlibrary{arrows.meta}

\input{adtfigures}

\usepackage{manfnt}

\newif\ifdraft\draftfalse
\ifdraft
\newcommand\modedraft[1]{#1}

\newcommand\todo[1]{{\color{purple}
    [\textbf{To do:} #1]}}
\newcommand\spcom[1]{\marginpar[{\color{blue}\small\dbend}]{\color{blue}\small\dbend}{\footnotesize \color{blue}[#1 - \textbf{Sophie}]}}
\newcommand\atcom[1]{\marginpar[{\color{orange}\small\dbend}]{\color{orange}\small\dbend}{\footnotesize \color{orange}[#1 - \textbf{Alexandre}]}}
\newcommand\tbcom[1]{\marginpar[{\color{red}\small\dbend}]{\color{red}\small\dbend}{\footnotesize \color{red}[#1 - \textbf{Thomas}]}}

\else
\newcommand\modedraft[1]{}
\newcommand\todo[1]{}
\newcommand\spcom[1]{}
\newcommand\atcom[1]{}
\newcommand\tbcom[1]{}

\fi
\title{Semantics of Attack-Defense Trees for Dynamic Countermeasures and  a New Hierarchy of Star-free Languages}
\titlerunning{Semantics of Attack-Defense Trees for Dynamic Countermeasures} 

 \author{Thomas Brihaye \inst{1}
 \orcidID{0000-0001-5763-3130} 
 \and
 Sophie Pinchinat \inst{2} \and
 Alexandre Terefenko \inst{1,2}}
\institute{University of Mons, Belgium  \\   \and
 University of Rennes, France}

\begin{document}
\maketitle

\begin{abstract}
We present a mathematical setting for attack-defense trees, a
classic graphical model to specify attacks and countermeasures. We equip attack-defense trees
 with (trace) language semantics allowing to have an
original dynamic interpretation of countermeasures. Interestingly, the
expressiveness of \adts coincides with star-free languages, and the
nested countermeasures impact the expressiveness of \adts. With an
adequate notion of countermeasure-depth, we exhibit a strict hierarchy of the
star-free languages that does not coincides with the classic
one. Additionally, driven by the use of
\adts in practice, we address the decision problems of trace
membership and of non-emptiness, and study their computational
complexities parameterized by the countermeasure-depth.
\end{abstract}
%\tableofcontents

\section{Introduction}
\label{sec-introduction}
Security is nowadays a subject of increasing attention as means to
protect critical information resources from disclosure, theft or
damage. The informal model of
\emph{attack trees} is due to Bruce Schneier\footnote{\url{https://www.schneier.com/academic/archives/1999/12/attack_trees.html}}
to graphically represent and reason about possible threats
one may use to attack a system. Attack trees have then been
widespread in the industry and are advocated since the 2008 NATO
report to govern the evaluation of the threat in risk analysis. The
attack tree model has attracted the interest of the
academic community in order to develop their mathematical theory
together with formal methods (see the survey \cite{widel2019beyond}).
 
Originally in \cite{schneier1999attack}, the model of attack tree aimed at describing
how an attack goal refines into subgoals, by using two
\emph{operators} $OR$ and $AND$ to coordinate those refinements. The
subgoals are understood in a ``static'' manner in the sense that
there is no notion of temporal precedence between them. Still, with
this limited view, many analysis can be conducted (see for example
\cite{kordy2014attack,gadyatskaya2016modelling}).
Then, the academic community considered two extensions of attack
trees. The first one, called \emph{\adt} (\shadt, for short), is obtained by augmenting
attack trees with nodes representing countermeasures
\cite{mauw2005foundations,kordy2012computational}. The second one, initiated by
\cite{DBLP:conf/sefm/PinchinatAV14,jhawar2015attack}, concerns a
``dynamic'' view of attacks with the ability to specify that the
subgoals must be achieved in a given order. This way to coordinate
the subgoals is commonly specified by using operator $SAND$ (for
Sequential $AND$). In
\cite{audinot2017my}, the authors proposed a \emph{path semantics} for
attack trees with respect to a given a transition system (a model
of the real system).
However, a unifying formal semantics amenable to the coexistence of
both extensions of attack trees -- namely with the defense and the dynamics
-- has not been investigated yet.

\renewcommand{\adt}{adt\xspace}
\renewcommand{\adts}{adts\xspace}
\renewcommand{\Adts}{Adts\xspace}

%\textbf{Contributions.}
In this paper, we propose a formal language semantics of \adts, in the
spirit of the trace semantics by \cite{brihaye2022adversarialgandalf}
(for defenseless attack trees), that allows countermeasure features via
the new operator $\Cop$ (for ``countermeasure'').  Interestingly, because
in \adts, countermeasures of countermeasures exist, we
define the \emph{countermeasure-depth} (maximum number of
nested $\Cop$ operators) and analyze its role in terms of
expressiveness of the model.

First, we establish the Small Model Property for \adts with countermeasure-depth
bounded by one (\Cref{boundedsizetracesADT1}), which ensure the existence of small traces in a non-empty semantics. This not so trivial result is a
stepping stone to prove further results.

Second, since our model of \adts is very close to \emph{\seres}
(\shseres for short), that are star-free regular expressions extended
with intersection and complementation, we provide a two-way
 translation from the former to the latter
(\Cref{theo:starfreeequalsADT}). It is known that the class of
languages denoted by \shseres coincides with the class
of \emph{star-free languages} \cite{rozenberg2012handbook}, that can
also be characterized as the class of languages definable in first-order logic over
strings (\FO). We make explicit a translation from \adts into \FO
(\Cref{lem:adttofo}) to shed light on the role played by the
countermeasure-depth. Our translation is reminiscent of the
constructions in \cite{PERRIN1986393} for an alternative proof of the
result in \cite{thomas1982classifying} that relates the classic dot-depth
hierarchy of star-free languages and the \FO quantifier
alternation hierarchy. In particular, we show (\Cref{lem:adtktofok+1})
that any language definable by an \adt with countermeasure-depth less than
equal to $k$ is definable in $\foalt[k+1]$, the $(k+1)$-th level of
the first-order quantifier alternation hierarchy.

\noindent Starting from the proof used in \cite{thomas1982classifying} to show
the strictness of the dot-depth hierarchy, we demonstrate that there
exists an infinite family of languages whose definability by an \adt
requires arbitrarily large countermeasure-depths. It should be noticed that
our notion of countermeasure-depth slightly differs from the
complementation-depth considered in \cite{stockmeyer1974complexity}
for \eres\footnote{arbitrary regular expressions extended with
intersection and complementation.}, because the new operator \ANDop is
rather a relative complementation. As a result, the countermeasure-depth
of \adts induces a new hierarchy of all star-free languages, that we
call the \emph{\ADT-hierarchy}, that coincides (at least on the very first levels) neither
with the dot-depth hierarchy, nor with the first-order logic quantifier
alternation hierarchy.

Third, we study three natural decision problems for
\adts, namely the \emph{membership problem} (\ADTmemb),
the \emph{non-emptiness problem} (\ADTNE) and the \emph{equivalence
problem} (\ADTequiv). The problem \ADTmemb is to determine if a
trace is in the semantics of an \adt. From a practical security point
of view, \ADTmemb addresses the ability to recognize an
attack, say, in a log file. The problem \ADTNE consists of, given an
\adt, deciding if its semantics is non-empty. Otherwise
said, whether the information system can be attacked or not. Finally,
the problem \ADTequiv consists of deciding whether two \adts describe
the same attacks or not. Our results are summarized in \Cref{table:results}.

The paper is organized as follows. \Cref{sec-example} proposes an introductory example. Next, we define our model of \adts in
\Cref{sec-adts}, their trace semantics and countermeasure-depth, and present the Small Model
Property for \adts with countermeasure-depth bounded by one. We then
show in \Cref{sec-expressiveness} that \adts coincide with star-free
languages. We next study the novel hierarchy induced by the
countermeasure-depth (\Cref{sec-ADThierarchy}) and study
decision problems on \adts (\Cref{sec-decisionproblems}).

\section{Introductory Example}
\label{sec-example}
Consider a thief (the proponent) who wants to steal two documents
inside two different safes (Safe 1 and Safe 2), without being seen.
The safes are located in two different but adjacent rooms (Room 1 and
Room 2) in a building; the entrance/exit door of the building leads to
Room 1.  The rooms are separated by a door and each room has a
window. Initially, the thief is outside of the building. A strategy
for the proponent to steal the documents is to attempt to open Safe 1
until it succeeds, then open Safe 2 until it succeeds (and finally to
exit the building). However, this strategy can be easily countered by
the company, say by hiring a security guard visiting the rooms on some regular
basis.

% Let us now assume that the company hires a guard as a countermeasure
% to prevent thefts. The guard is asked to turn persistently around the
% building. %(at a constant speed).
% During her tour, the guard faces the windows of the rooms and 
% checks whether the corresponding room is empty.
%

% \begin{figure}[h]
% \begin{center}
% \input{tikz-building}
% \end{center}
% \caption{Building plan}
% \label{fig-building}
% \end{figure}

% We assume that  the guard takes two time units to go from one
% window to the other one.  We also assume that a proponent's
% attempt to open a safe takes one unit of time (but this attempt may
% not succeed). Clearly, the aforementioned strategy for the thief is no longer
% successful as she may be seen by the guard while staying too long in
% the same room.

Security experts would commonly use an \adt to describe how the
proponent may achieve her goal and, at the same time, the ways
its opponent (the company) may prevent the proponent from reaching her
goal. An informal
\adt expressing the situation is given in
Figure~\ref{fig-adt-countermeasure}, where traditionally goals of
the proponent are represented in red circles, while countermeasures
of the opponent are represented in green squares. An arrow from a left
sibling to a right sibling specifies that the former goal must be
achieved before starting the latter. A countermeasure targeting a proponent goal
is represented with dashed lines.

%% \tbcom{Dire que c'est la représentation
%% classique dans la literature, mais que nous, ce n'est pas ça qu'on va
%% utiliser. Veut-on dans l'intro donner notre représentation de l'arbre?
%% Je ne sais pas trop comment je ferai...} The ellipsoidal note
%% represents the main objective of the thief and the rectangle note a
%% countermeasure implemented by the company in order to prevent the
%% attack.

\begin{figure}[h]
\begin{subfigure}[b]{0.5\textwidth}         
\centering
\includegraphics[scale=.9]{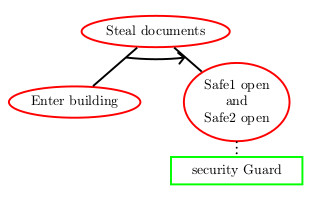}
\caption{A countermeasure from the Company.}
\label{fig-adt-countermeasure}
\end{subfigure}
\hfill
\begin{subfigure}[b]{0.5\textwidth}
\centering
\includegraphics[scale=.9]{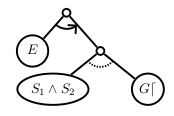}
\caption{Formal representation of the \adt}
\label{fig-adt-formalcountermeasure}
\end{subfigure}
\caption{\Adts for the thief and company problem.}\label{fig-adts-intro}
\end{figure}

% Unfortunately for the company, this is not the end of the story as the
% proponent can still reach her objective, namely open the two
% safes. For instance, she could conduct her attempts to open Safe 1 and
% Safe 2 by alternating them at each unit of time, thus exhibiting a
% countermeasure to the countermeasure of the opponent, as depicted by
% the \adt of Figure~\ref{fig-adt-countermeasureofcountermeasure}.

%% \tbcom{Veut-on déjà parler de la syntaxe et sémantique ici ou c'est mieux d'attendre plus tard?}

As said, the graphical model of \Cref{fig-adt-countermeasure} is informal and
cannot be exploited by any automated tool for reasoning. With the
setting proposed in this contribution, we make it formal, and in
particular we work out a new binary operator \Cop for
``countermeasure'', graphically reflected with a curved dashed line
between two siblings: the left sibling is a proponent's goal while the
right sibling is the opponent's countermeasure -- with this convention,
we can unambiguously retrieve the player's type of an \adt node. The
more formal version of the \adt in \Cref{fig-adt-countermeasure} is
drawn in \Cref{fig-adt-formalcountermeasure} (details
for its construction can be found in \Cref{ex:introductory}).

The proposed semantics for \adts also allows us to consider
nested \Cop operators to express countermeasures of
countermeasures. For example, a proponent countermeasure against the
company countermeasure could be to be disguised as an employee working
in the building; we formalise this situation
in \Cref{ex:introductory}. It should be observed that nested
countermeasure is a core aspect of our contribution and 
the main subject of \Cref{sec-ADThierarchy}.

\section{Attack-Defense Trees and Countermeasure-Depth}
\label{sec-adts}
\paragraph{Preliminary notations}
\label{sec-background}
%\spcom{check everything needed is here}

For the rest of this paper, we fix $\Prop$ a finite set of
propositions and we assume that the reader is familiar with
propositional logic. We use typical symbol $\formula$ for
propositional formulas over $\Prop$ and write a valuation of the
propositional variables $\val$ as an element of $\alphabet\egdef
2^{\Prop}$, that will be viewed as an alphabet. A \emph{trace}
$\trace$ over $\Prop$, is finite word over $\alphabet$, that is a
finite sequence of valuations. We denote the empty trace by $\eword$,
and we define $\Tracesne:=\Traces\backslash \{\eword\}$. For a trace
$\trace=\val_1...\val_n$, we define $\sizetrace{\trace}:=n$, its
the \emph{length}, as the number of valuations appearing in $\trace$,
and we let $\trace[i]\egdef\val_i$, for each
$i \in \{1,\ldots,n\}$. We define the classic \emph{concatenation} of
traces: given two traces $\trace=\val_1...\val_n$ and
$\trace'=\val_1'...\val'_m$, we define
$\trace\concat\trace':=\val_1...\val_n\val_1'...\val_m'$. We also lift
this operator to sets of traces in the usual way: given two sets of
traces $\traceset$ and $\traceset'\subseteq \Traces$, we let
$\traceset\concat\traceset'\egdef\{\trace\concat\trace'~:~\trace\in\traceset \text{
and } \trace'\in\traceset'\}$. For a trace
$\trace=\val_1...\val_n \in \Traces$ and $1\leq i \leq n$, 
the trace $\trace'=\val_1...\val_i$ is a \emph{prefix} of $\trace$,
written $\trace' \prefix \trace$.\\ 

\label{sec-definition}

%
%% and briefly comment on the
%% proposed definition w.r.t.\ other ones in the literature (see the
%% survey \cite{widel2019beyond}).
%

%\subsection{Syntax of \adts}
%------------------------
%\label{sec-syntax}
We define attack-defense trees (\adts) over $\Prop$, as well as their \emph{trace semantics} and
their \emph{countermeasure-depth}, and develop enlightening examples.
\Adts are standard labeled finite trees with a dedicated set of labels based on the special $\epsilon$ label and propositional formulas for leaves and on the set $\{\ORop, \SANDop, \ANDop, \Cop\}$ for internal nodes.

\begin{definition}\label{def:ADT}
The set \ADT of \emph{\adts} over $\Prop$ is inductively defined by:
\begin{itemize}
  \item the \emph{empty-word leaf} $\epsilonadt$ and every
    propositional formula $\formula$ over $\Prop$ are in \ADT;
\item if trees $\att_1, ...,\att_n$ are in $\ADT$, so are $\ORop(\att_1,...,\att_n)$, $\SANDop(\att_1,...,\att_n)$, and $\ANDop(\att_1,...,\att_n)$;
\item if trees $\att$ and $\att'$ are in $\ADT$, so is $\Cop(\att,\att')$.
\end{itemize}

The \emph{size} of an \adt $\att$, written $\sizeatt{\att}$ is defined
as the sum of the sizes of its leaves, provided the size of
$\epsilonadt$ is $1$, while the size of $\formula$ is its size when
seen as a propositional formula.
%We let $\ADTks[]$ be the subset of \ADT where we discard the empty-word leaves $\epsilonadt$.
\end{definition}

% Definition~\ref{def:ADT} differs from classical definitions from the literature \tbcom{donner une ou deux refs} in two points \tbcom{reformulation OK?}. Firstly, we do not strongly distinguish two types of nodes: "proponent node" and "opponent node". This is not a real issue because, in classic definitions we ask that the root is a proponent node, then the type of all other nodes are determined by the type of its parent.  Secondly, we use a new type of leaf: $\epsilonadt$. The interpretation for this leaf is: "doing nothing is the only way to attain the objective". We will show in \cref{sec-RemarkableADT} how we use this leaf to have more conciseness in our attack-defense trees. Moreover, we discuss in \cref{sec-decisionProblems} that using empty leaves does not change any complexity results. 

% \tbcom{Ci-dessous une reformulation du paragraphe précédent qui doit être complétée}

Regarding the semantics, \adts describe a set of traces over
alphabet $\alphabet\egdef 2^\Prop$, hence their \emph{trace
semantics}. Formally, for an \adt $\att$, we define the language
$\semword{\att} \subseteq \Traces$. First, we set
$\semword{\epsilonadt}=\{\epsilon\}$. Now, a leaf \adt hosting formula
$\formula$ denotes the reachability goal $\formula$, that is the set
of traces ending in a valuation satisfying $\gamma$ (we use the classic notations $\true=\prop \lor \lnot \prop$ and $\false=\lnot \true$ with $\prop\in \Prop$). We make our trace
semantics compositional by providing the semantics of the four
operators \ORop, \SANDop, \Cop, and \ANDop in terms of how the
subgoals described by their arguments interact. Operator \ORop
tells that at least one of the subgoals has to be achieved. Operator
\SANDop requires that all the subgoals need being achieved in the
left-to-right order. The binary operator \Cop requires to achieve the
first subgoal without achieving the second one. Finally, operator
\ANDop tells that all subgoals need to be achieved, regardless of the
order. Without any countermeasure, \ANDop can be seen as a relaxation
of the \SANDop, but it is not true in general (see
\cref{ex:sandeach}). % Finally, operator \ANDop is a relaxation of operator \SANDop in the sense that the order in which the subgoals are achieved does not matter, as long they are all achieved.

At the level of the property described by an \adt, i.e.\ a trace
language, the operators correspond to specific language operations:
\ORop corresponds to union, \SANDop to concatenation, \Cop to a
relativized complementation. Only \ANDop corresponds to a
less classic operation: a trace $\trace$ belongs to the language of $\ANDop(\att_1,\att_2)$ if $\trace$ belongs to the language of $\att_1$ and has a prefix in the language of $\att_2$ or vice-versa. Formally:

% \begin{example}
%   \label{ex:semanticsand}
% The \adt $\ANDop(\formula_1,\formula_2)$ denotes the set of traces
% $\trace$ where one of the formulas $\formula_i$'s is reached at the
% end of $\trace$ and while the other formula is reached at some point
% along $\trace$.
% \end{example}

\begin{definition}
\label{def:andword}
Let $\wordlang_1$, $\wordlang_2$ be two languages over alphabet
$\alphabet$. The \emph{each} of $\wordlang_1$ and $\wordlang_2$ is
the language $\wordlang_1 \andword \wordlang_2 \egdef (\wordlang_1
\concat \alphabet^* \inter \wordlang_2) \union (\wordlang_2 \concat \alphabet^*\inter
\wordlang_1 )$.
%% $\wordlang_1 \andword \wordlang_2 =\{\word_1 \in \wordlang_1 |$ there
%% exists $ \word_2 \in \wordlang_2 $ such that $\word_2\prefix
%% \word_1\} \union \{\word_2 \in \wordlang_2 |$ there exists $ \word_1
%% \in \wordlang_1 $ such that $\word_1\prefix \word_2\}$.
\end{definition}

Because operator $\andword$ is associative, we can define
$\wordlang_1 \andword \wordlang_2 \andword \ldots \andword
\wordlang_n$ that amounts to being equal to $\UNION_{i \in
  \{1,\ldots,n\}} (\wordlang_i \inter \INTER_{j\neq i}\wordlang_j
\concat \alphabet^*)$.

\begin{example}
  \label{ex:andword}
  %\spcom{can we do better}
  A word $\word=\letter_1\letter_2\ldots\letter_m$ belongs to
  $\wordlang_1 \andword \wordlang_2 \andword \wordlang_3$ whenever
  there are three (possibly equal) positions $i_1,i_2,i_3=m$ such that, for each $j \in \{1, 2, 3\}$,
  the word prefix $\letter_1\ldots\letter_{i_j} \in
  \wordlang_{\pi(j)}$, for some permutation $\pi$ of $(1,2,3)$. 
\end{example}
%\subsection{Semantics}
%------------------------
%\label{sec:semantics}

We can now formally define the \adt semantics.%\newpage 
\renewcommand{\alphabettraces}{\alphabet}
\begin{definition}
  \label{def:semword}~
\begin{itemize}
\item $\semword{\epsilonadt}\egdef\{\eword\}$ and $\semword{\formula}\egdef \{\val_1...\val_n \in \Traces :  \val_n \models \formula\}$;\\ In particular, $\semword{\true}=\Tracesne$ and $\semword{\false}=\emptyset$;
\item $\semword{\ORop(\att_1, ..., \att_n)}\egdef\semword{\att_1} \union ... \union \semword{\att_n}$;
\item $\semword{\SANDop(\att_1, ..., \att_n)}\egdef\semword{\att_1} \concat ... \concat \semword{\att_n}$;
\item $\semword{\Cop(\att_1, \att_2)}\egdef \semword{\att_1} \setminus \semword{\att_2}$;%, or equivalently $\semword{\att_1} \inter \complement{\semword{\att_2}}$.
\item $\semword{\ANDop(\att_1, ..., \att_n)}=\semword{\att_1} \andword ... \andword \semword{\att_n}$.
\end{itemize}
\end{definition}

In the rest of the paper, we say for short that \emph{an \adt is non-empty}, written $\att \neq \emptyset$, whenever
$\semword{\att}\neq\emptyset$. We say that two \adts $\att$ and $\att'$ are
\emph{equivalent}, whenever $\semword{\att}=\semword{\att'}$.

\begin{remark}
  \label{rem:associative}
Since all operators $\union$, $\concat$ and $\andword$ over trace
languages are associative, the trees of the form $\OP(\att_1,
\OP(\att_2, \att_3))$, $\OP(\OP(\att_1,\att_2), \att_3)$, and
$\OP(\att_1, \att_2, \att_3)$ are all equivalent, when $\OP$ ranges
over $\{\ORop, \SANDop, \ANDop\}$.
As a consequence, we may sometime assume that nodes with such
operators are binary.
\end{remark}

We now introduce some notations for particular \adts to ease our
exposition and provide some examples of \adts with their
corresponding trace property.

We define a family of \adts of the form $\adtge, \adtle,
\adteq$, where $\ell$ is a non-zero natural. We let $\adtge \egdef
\SANDop(\true,...,\true)$ where $\true$ occurs $\ell$ times;
$\adtle\egdef\Cop(\true, \adtge)$; and $\adteq\egdef\Cop(\adtge,
\adtge[\ell+1])$.  It is easy to establish that \adt $\adtge$
(resp.\ $\adtle$, $\adteq$) denotes the set of traces of length at
least (resp.\ at most, exactly) $\ell$.
%% $$\begin{array}{lcl}
%%   \semword{\adtge}&=&\{\trace \in \Traces \,|\, \ell \leq \sizetrace{\trace}\}\\
%%   \semword{\adtle}&=&\{\trace \in \Traces \,|\, \sizetrace{\trace} < \ell\}\\
%% \semword{\adteq}&=&\{\trace \in \Traces \,|\, \sizetrace{\trace}= \ell\}
%% \end{array}
%% $$
%
%% \begin{proof} The proof is a mere application of the semantics.
%% Remark that $\semword{\true}=\Tracesne$, indeed, each trace $\trace \in \semword{\true}$ finishes with a valuation $\val$ such that $\val \models \true$, in other words, $\trace \neq \eword$. Therefore we have that $\semword{\adtge{i}}=(\Tracesne)^i=\{\trace \in \Traces | \sizetrace{\trace}\leq i\}$. Moreover, $\semword{\adtle{i}}=\Tracesne\backslash \{\trace \in \Traces | \sizetrace{\trace}\leq i\}= \{\trace \in \Traces | \sizetrace{\trace} < i\}$. Finally $\semword{\adteq{i}}= \{\trace \in \Traces | \sizetrace{\trace}\leq i\} \backslash \{\trace \in \Traces | \sizetrace{\trace}\leq i+1\} = \{\trace \in \Traces | \sizetrace{\trace}= i\}$.
%% \end{proof}
%
We also consider particular \adts and constructs for them.
\begin{itemize}
\item $\etrue \egdef\ORop(\epsilonadt,\true)$,
\item $\coatt \egdef \Cop(\etrue,\att)$,\label{def:coatt}
and $\andatt{\att_1}{\att_2} \egdef \coatt[\ORop(\coatt[\att_1],\coatt[\att_2])]$,
\item $\allattall\egdef \SANDop(\etrue,\att,\etrue)$, $\allatt\egdef \SANDop(\etrue,\att)$ and $\attall\egdef \SANDop(\att,\etrue)$.
\item Given a formula $\formula$ over $\Prop$, we let
$\strict{\formula}:=\Cop(\formula, \adtge[2])$. 
\end{itemize}

Based on these notations, we develop further examples.

\begin{example}
  \label{ex:adtsexamples}
  \begin{itemize}
  \item $\semword{\etrue}=\Traces$;
  \item $\semword{\coatt}=\Traces\setminus \semword{\att}$;
  \item $\semword{\andatt{\att_1}{\att_2}}=\semword{\att_1} \inter
    \semword{\att_2}$;
    \item $\semword{\strict{\formula}}$ is set of one-length
      traces whose unique valuation satisfies $\formula$; in
      particular, when a valuation $\val$ is understood as a
      formula, namely formula $\bigwedge_{\prop\in\val} \prop \land
      \bigwedge_{\prop \not \in \val} \lnot \prop$, the \adt
      $\strict{\val}\egdef\Cop(\val, \adtge[2])$ is such that
      $\semword{\strict{\val}}=\{\val\}$;

    \item For a valuation $\val$,
      \begin{itemize}
        \item
          $\semword{\allattall[\strict{\val}]}=\semword{\allattall[\val]}=\semword{\attall[\val]}=\Traces\val\Traces$;
          \item 
      $\semword{\allatt[\strict{\val}]}=\semword{\allatt[\val]}=\semword{\val}=\Traces\val$;
    \item also, $\semword{\attall[\strict{\val}]}=\val\Traces$.
      \end{itemize}
\end{itemize}
\end{example}

 \begin{example}
\label{ex:sandeach} 
 Note that $\ANDop$ cannot be seen as a kind of relaxation of $\SANDop$. For the set of propositions $\{\prop\}$,  if we consider the formula $\prop$ as a leaf, $\semword{\strict{\prop}}=\{\prop\}$ and $\semword{\strict{\lnot\prop}}=\{\emptyset\}$. Thus $\SANDop(\strict{ \lnot\prop},\strict{\prop})=\{\trace \}$ with trace $\trace=\emptyset\prop$. However $\ANDop(\strict{ \lnot\prop},\strict{\prop})=\emptyset$. Let us notice that the construction of this example uses the $\Cop$ operator (hidden in $\strict{\prop}$ and $\strict{\neg \prop}$).
\end{example}       
\begin{example}    
\label{ex:introductory}
We come back to the situation of our introductory example (\Cref{sec-example}). First, we discuss the formal semantics of the informal tree in \cref{fig-adt-countermeasure}. To do so, we propose the following set of propositions: $\Prop=\{E,S_1,S_2,G\}$ where $E$ holds when the thief is entering the building, $S_1$ (resp. $S_2$) holds when the first (resp. second) safe is open, and $G$ is true if a guard is in the building. The situation can be described by the following \adt: $\att_{ex_1}=\SANDop(E, \Cop(S_1\land S_2, \attall[G]))$, represented in \Cref{fig-adt-formalcountermeasure}, where we distinguish \SANDop with a curved line and \Cop with a dashed line. We have $\semword{\att_{ex_1}}=\{\val_1...\val_n\in \Traces ~:~ \val_n \models E\} \concat \big(\{\val_1...\val_n\in \Traces ~:~ \val_n \models S_1 \land S_2\} \setminus \{\val_1...\val_n\in \Traces ~:~ \exists i$ such that $ \val_i  \models G\}\big)$. If we write $\gamma_\varphi=\{\val \in 2^\Prop:\val \models \varphi\}$, we have $\semword{\att_{ex_1}}=\Traces \gamma_E \cdot (\Traces \gamma_{S_1 \wedge S_2} \setminus \Traces \gamma_G  \Traces)$.
In other words, we want all traces where $E$ holds at some point and, after it, $G$ cannot be true and finish by a valuation where $S_1\land S_2$ holds.

In order to illustrate the nesting  of countermeasures, we now allow the thief to disguise himself as an employee (assuming that when disguised, the guard does not identify him as a thief). To do so, we extended the set of propositions: $\Prop'=\{E,S_1,S_2,G,D\}$, where $D$ holds when the thief is disguised. The situation is now described by the following \adt: $\att_{ex_2}=\SANDop(E, \Cop(S_1\land S_2,\allattall[\Cop(G,D)]))$. The semantics for $\att_{ex_2}$ is all traces where $E$ holds at some point and, after it, $G$ cannot be true, except if $D$ holds at the same time, and finish by a valuation where $S_1\land S_2$ holds. A representation of $\att_{ex_2}$ can be found in \Cref{app:fig-countermeasureofcountermeasure} 
\end{example}

\label{sec-subclasses}
%% \label{sec-ADTk}

We now stratify the set \ADT of \adts
according to their \emph{countermeasure-depth} that denotes the
maximum number of nested countermeasures.

%% In the remaining of this section, we stratify the set \ADT of \adts
%% according to their \emph{countermeasure-depth} that denotes the
%% maximum number of nested countermeasures.
\begin{definition}
The \emph{countermeasure-depth} of an \adt $\att$, written $\counterdepth(\att)$, is 
%the natural 
inductively defined by:
\begin{itemize}
    \item $\counterdepth(\epsilonadt)\egdef\counterdepth(\formula)=0$;
    \item $\counterdepth(\OP(\att_1, ..., \att_n))\egdef\max\{\counterdepth(\att_1), ..., \counterdepth(\att_n)\}$ for every $\OP\in \{\ORop, \SANDop, \ANDop\}$;
    \item $\counterdepth(\Cop(\att_1, \att_2))\egdef\max\{\counterdepth(\att_1), \counterdepth(\att_2)+1\}$
\end{itemize}
\end{definition}
%\begin{definition}
We let $\ADTk\egdef\{\att \in \ADT~:~ \counterdepth(\att)\leq k\}$ be the set of \adts with countermeasure-depth at most $k$.
%\end{definition}
Clearly $\ADTk[0] \subseteq \ADTk[1] \subseteq \ldots \ADTk \subseteq
\ADTk[k+1] \subseteq \ldots$, and  
$\ADT=\underset{k\in\setn}{\bigcup}\ADTk$.
\begin{example}
  \label{ex:counterdepth} We list a couple of
examples. $\cdepth(\coatt)=1+\cdepth(\att)$;
$\cdepth(\andatt{\att_1}{\att_2})=2+\max\{\cdepth(\att_1),\cdepth(\att_2)\}$;
$\cdepth(\Cop(\Cop(\formula_1,\formula_2), \formula_3))= 1$ while
$\cdepth(\Cop(\formula_3, \Cop(\formula_2,\formula_3)))= 2$;
$\cdepth(\adtle)=\cdepth(\adteq)=1$ while $\cdepth(\adtge)=0$;
$\cdepth(\strict{\formula})=1$; In \Cref{ex:introductory},
$\cdepth(\att_{ex_1})=1$ and $\cdepth(\att_{ex_2})=2$.
%
  %% \begin{itemize}
  %% \item $\cdepth(\coatt)=1+\cdepth(\att)$;
  %%     \item $\cdepth(\andatt{\att_1}{\att_2})=2+\max\{\cdepth(\att_1),\cdepth(\att_2)\}$;
    
  %%       \item $\cdepth(\Cop(\Cop(\formula_1,\formula_2), \formula_3))= 1$
  %%     while $\cdepth(\Cop(\formula_3, \Cop(\formula_2,\formula_3)))= 2$;
  %%         \item $\cdepth(\adtle)=\cdepth(\adteq)=1$ while $\cdepth(\adtge)=0$;
  %%         \item $\cdepth(\strict{\formula})=1$;
  %%         \item In \Cref{ex:introductory}, $\cdepth(\att_{ex_1})=1$ and $\cdepth(\att_{ex_2})=2$
  %% \end{itemize}
  %% \tbcom{virer itemize pour gagner de la place?}
%%\end{example}

%

%%\begin{example}
  \label{ex:quelADT} 
Also, $\counterdepth(\adtge)=0$, $\counterdepth(\adtle)=\counterdepth(\adteq)=1$, so that $\adtge\in \ADTk[0]$;
$\adtle$ and $\adteq \in \ADTk[1]$. Moreover, $\att_{ex_1}\in \ADTk[1]$ and $\att_{ex_2}\in \ADTk[2]$.
\end{example}

%% Because we aim at studying the expressive power of \adts in terms of the (trace) languages they capture, we introduce the classic concept of \emph{definability}.
%% \begin{definition}
%%   \label{def-ADTkdefinability}
%% A language $\wordlang$ is \emph{$\ADTk$-definable} (resp.\ $\ADT$-definable), written
%% $\wordlang \in \ADTk$ (resp.\ $\ADT$), whenever $\semword{\att} = \wordlang$, for some $\att \in \ADTk$ (resp.\ for some $k$). 
%
%% Whenever a language $\wordlang$ is \ADT-definable, we let its \emph{countermeasure-depth} be the least $k \in \setn$ such that $\wordlang \in \ADTk$.\spcom{might be useful to align with the dot-depth of a language}.

We say that a language $\wordlang$ is \emph{$\ADTk$-definable} (resp.\ $\ADT$-definable), written
$\wordlang \in \ADTk$ (resp.\ $\ADT$), whenever $\semword{\att} = \wordlang$, for some $\att \in \ADTk$ (resp.\ for some $k$).

It can be established that non-empty \adts in $\ADTk[1]$ enjoy small
traces, i.e. smaller than the size of the tree.

\begin{theorem}[Small model property for \text{$\ADTk[1]$}]
\label{boundedsizetracesADT1}
An \adt $\att\in \ADTk[1]$ is non-empty if, and only if,
there is a trace $\trace \in \semword{\att}$ with
$\sizetrace{\trace} \leq \sizeatt{\att}$.
\end{theorem}

We here only sketch the proof, whose details can be found
in \Cref{app:smallmodelproperty}. The
technique we employ consists in defining a slight variant of the
classic relation of super-word in language theory, that we call the  \emph{lift} binary
relation. We prove that if $\att \in \ADTk[1]$, then there exists a finite set of generators, denoted $\gen{\att}$, which is sufficient to describe $\semword{\att}$ through the lift relation. Next, we can prove
that the traces in $\gen{\att}$ have size bounded by the number
of leaves of $\att$. Notice that the result also holds for $\ADTk[0]$,
since $\ADTk[0]\subseteq\ADTk[1]$. 

%With the lift
%relation, we can define the \emph{generators} of a trace language,
%that may not exist in general. However, we are able to
%inductively define the set $\gen{\att}$ as a finite set of traces, and show
%that for $\att \in \ADTk[1]$, the set $\gen{\att}$ is indeed a set of
%generators of $\semword{\att}$. 
%Next, we can prove
%that the traces in $\gen{\att}$ have size bounded by the number
%of leaves of $\att$. Notice that the result also holds for $\ADTk[0]$,
%since $\ADTk[0]\subseteq\ADTk[1]$.

\section{\Adts, Star-free Languages, and First-Order Logic}
\label{sec-expressiveness}

%This section is dedicated to the proof that \adts coincide with star-free languages and first order formulas.
We prove that \adts coincide with star-free languages and first order formulas.

\subsection{Reminders on Star-free Languages and First-Order Logic}

The class of \emph{star-free languages}
introduced
by \cite{mcnaughton1971counter,eilenberg1974automata,pin1984varietes}
(over alphabet $\alphabet$) is obtained from the
finite languages (or alternatively languages consisting of a single
one-length word in $\alphabet$) by finitely many applications of
Boolean operations ($\union$, $\inter$ and $\complement*$ for the
complement) and the concatenation product (see \cite[Chapter
7]{rozenberg2012handbook}). Alternatively, one characterizes star-free
languages by first considering \eres\ -- that are regular expressions
augmented with intersection and complementation, and second by
restricting to \emph{\seres} (\shseres, for short) that are \eres with
no Kleene-star operator.
Regarding computational complexity aspects, we recall the following
the subclass of \shseres of \eres. The word membership problem (\ie
whether a given word belongs to the language denoted by a \shsere) is
in \PTIME \cite[Theorem 2]{kupferman2002improved}, while the
non-emptiness problem (\ie is the denoted language empty?) and the
equivalence problem (\ie do two \shseres denote the same language?)
are hard, both non-elementary \cite[p.\ 162]{stockmeyer1974complexity}.
%% \begin{theorem}
%% For \shseres:
%% \begin{enumerate}
%%     \item 
%% \label{theo:MEMstarfree}
%% The word membership problem  is in \PTIME.(see {\cite[Theorem 2]{kupferman2002improved}})
%% \item 
%% \label{theo:NEandEQUIVstarfree}
%% The non-emptiness and the equivalence problems  are non-elementary.(see {\cite[p.\ 162]{stockmeyer1974complexity}})
%% \end{enumerate}
%% \end{theorem}

We now recall classical results on the first-order
logic on finite words \FO (see details in \cite[Chapter
29]{DBLP:books/ems/21/P2021}). The signature of \FO,
say for words over an alphabet $\alphabet$, is composed of a unary
predicate $a(x)$ for each $a \in \alphabet$, whose meaning is the ``letter
at position $x$ of the word is $a$'', and the binary predicate $x<y$
that states ``position $x$ is strictly before position $y$ in the
word''. For a \FO-formula, we define its \emph{size}  $\sizeformula{\foformula}$ as the size of the expression considered as a word.
A language $\wordlang$ is \FO-\emph{definable} whenever there exists a
\FO-formula $\foformula$ such that a word $\word \in \wordlang$ if,
and only if, $\word$ is a model of $\foformula$. Similarly, we say that an \adt $\att$ is \FO-definable if $\semword{\att}$ is \FO-definable.
It is well-known that \FO-definable languages coincide with star-free languages   \cite{mcnaughton1971counter,thomas1982classifying,PERRIN1986393}.
%% \cite{schiering1996counter} (see also the recent handbook \cite[Chapter 16]{DBLP:books/ems/21/P2021}).

%% \begin{theorem}[{\cite{mcnaughton1971counter,thomas1982classifying,PERRIN1986393}}] \label{fobijstarfree}
%% A language $\wordlang$ is star-free if, and only if, it
%% is \FO-definable.
%% \end{theorem}

Also, for a fine-grained inspection of \FO, let us denote by
$\foalt[\ell]$ (resp. $\foaltdual[\ell]$) the fragments of \FO
consisting of formulas with at most $\ell$ alternation of $\exists$
and $\forall$ quantifier blocks, starting with $\exists$
(resp. $\forall$). %\tbcom{add small example}
The folklore results
regarding satisfiability
of \FO-formulas \cite{stockmeyer1974complexity,meyer2006weak} are:
\begin{inparaenum}[(a)]
\item The satisfiability problem for \FO is non-elementary;
\item The satisfiability for $\foalt[\ell]$ is in $(\ell-1)$-\EXPSPACE.\footnote{with the convention that $0$-\EXPSPACE=\PSPACE.}
\end{inparaenum}
%
%% \begin{theorem}[{\cite{stockmeyer1974complexity,meyer2006weak}}]~
%% \label{theo:complexityfoalt}
%% \begin{itemize}
%% \item The satisfiability problem for \FO is non-elementary;
%% \item The satisfiability for $\foalt[\ell]$ is in $(\ell-1)$-\EXPSPACE. \footnote{with the convention that $0$-\EXPSPACE=\PSPACE.}
%% %% \begin{itemize}
%% %% \item $\foalt[1]$ is in \PSPACE; 
%% %% \item  $\foalt[\ell]$ is $(\ell-1)$-\EXPSPACE-complete, for $\ell \geq 2$.
%% %% %\item $\foaltdual[\ell]$ is $\ell$-\EXPSPACE-complete. \spcom{likely useless}.
%% %% \end{itemize}
%% \end{itemize}
%% \end{theorem}
%
We are not aware of any result that establishes a tight lower bound
complexity for the satisfiability problem on the $\foalt[\ell]$
fragments of \FO.

We lastly recall the definition of the \emph{dot-depth hierarchy} of star-free languages:
level $0$ of this hierarchy is $\dotdepth[0]\egdef\{\lang \subseteq 2^\alphabettraces : \lang$ is finite
or co-finite\}, and evel $\ell$ is
$\dotdepth \egdef\{\lang \subseteq 2^\alphabettraces
: \lang$ is a Boolean combination of languages of the form
$\lang_1\concat...\concat\lang_n$ where $\lang_1,
..., \lang_n \in \dotdepth[\ell-1]\}$. The dot-depth hierarchy has a tight connection with \FO fragments \cite{thomas1982classifying}: for every $\ell>0$, $\foalt[\ell] \subseteq \dotdepth \subseteq \foalt[\ell+1]$.

%% We show in the next theorem the
%% correlation between the dot-depth hierarchy and the quantifier
%% alternation hierarchy.
%% \begin{theorem}
%%     For each $\ell>0$, the following sets inclusion holds: $\foalt[\ell] \subseteq \dotdepth \subseteq \foalt[\ell+1]$
%% \end{theorem}

\subsection{Expressiveness of \Adts }

The first result of this section consist in showing that \adts and \seres share the same expressiveness.

\begin{theorem}
\label{theo:starfreeequalsADT}
A language $\wordlang$ is star-free if, and only if, $\wordlang$ is $\ADT$-definable.
\end{theorem}

%% We give a high-level argument for \Cref{theo:starfreeequalsADT},
%% before going through the mathematics.
%
For the ``only if'' direction of \Cref{theo:starfreeequalsADT}, we
reason by induction on the class of star-free languages. For a
language of the form $\{\val\}$ where $\val \in \alphabettraces$ one
can take the \adt $\strict{\val}$ (that is $\Cop(\val, \adtge[2])$).  Now we
can inductively build adequate \adts for compound star-free languages
by noticing that language operations of union and concatenation are
captured by \adts operators \ORop and \SANDop respectively, while complementation 
and intersection are obtained from the $\coatt[.]$ and
$\andatt{.}{.}$ as formalized in \Cref{ex:adtsexamples}. 
%
%% For each regular star-free expression $\reg$ over $2^{\Prop}$, we
%% associate an \adt, written $\att_\reg$, inductively as follow:
%% $\att_\emptyset=\false$; $\att_\val=\strict{\val}$;
%% $\att_{\reg \cup \reg'}=\ORop(\att_\reg, \att_{\reg'})$; 
%% $\att_{\reg . \reg'}=\SANDop(\att_\reg, \att_{\reg'})$; and
%% $\att_{\minus \reg}= \Cop(\true, \reg)$.
%% %\end{definition}
%% One easily proves by induction over $\reg$ that the translation is correct, namely that:
%% \begin{lemma}
%% \label{starfreetoADT}
%% %For every \ere $\reg$ over $2^{\Prop}$, we have
%% $\semword{\att_\reg}=\lang(\reg)$.
%% \end{lemma}
%
%% \begin{remark}
%% \label{starfreetoADT}
One easily verifies that %when the star-free language is given as an \shsere $\reg$, it is clear
that the size of the \adt corresponding to an \shsere $\reg$ is in $O(\sizereg{\reg})$, where $\sizereg{\reg}$ denotes the size (number of characters) of $\reg$.
%% one can easily
%% effectively compute an \adt $\att_{\reg}$ with
%% $\semword{\att_\reg}=\lang(\reg)$ and whose size is in
%% $O(\sizereg{\reg})$.
%% define an effective translation into \adts: 
%% Remark that $\sizeatt{\att_{\reg}} \in O(\sizereg{\reg})$. Indeed, we
%% have $\sizeatt{\emptyset}=1$,
%% $\sizeatt{\att_\val}=\sizeatt{\strict{\val}}=3$,
%% $\sizeatt{\att_{\reg \cup \reg'}}=\sizeatt{\att_{\reg
%% . \reg'}}=\sizeatt{\att_{\reg}}+\sizeatt{\att_{\reg'}}$ and
%% $\sizeatt{\att_{\minus \reg}}= \sizeatt{\att_{\reg}}+\spch{2}$.
%% \end{remark}
%% \begin{proof}
%% For an \ere of the shape $\val \in 2^{\Prop}$, we have
%% clearly that $\val$ is the only valuation satisfying formula
%% $\underset{\prop\in\val}{\bigwedge} \prop\land \underset{\prop \not \in \val}{\lnot \prop}$,
%% thus from \Cref{strictADT}, construction for $\att_\val$ is
%% correct. For other regular expressions, we can straightforwardly use
%% the definition of semantics of attack-defense trees.
%% \end{proof}

%\Cref{starfreetoADT} makes the ``only if'' part of \Cref{theo:starfreeequalsADT}.

\renewcommand{\ere}{\shere\xspace}
\renewcommand{\eres}{\sheres\xspace}
\renewcommand{\Eres}{\shEres\xspace}
%In the following, we shortly write \ere for extended regular expression.

For the ``if'' direction of \Cref{theo:starfreeequalsADT}, it is easy
to translate an \adt into an \shsere: the leaf $\epsilonadt$
translates into $\epsilon$, a leaf \adt $\formula$
translates into $\bigcup_{\val \models \formula} \val$ -- notice that this translation is exponential. For
non-leaf \adts, since every operator occurring in the \adt has its language-theoretic counterpart
the translation goes smoothly. However, the translation is exponential because of the \adt operator $\ANDop$, see \Cref{def:andword}.  \\

We now dig into the \ADT-hierarchy induced by the countermeasure-depth
and compare it with the \FO fragments $\foalt[\ell]$ and $\foaltdual[\ell]$.

We first step design a translation from \ADT into \FO, inductively over \adts. The translation of an \adt $\att$ is written $\foformula_\att$.
%% For the rest of this section, we try to be more precise. To do so, we present an effective translation of \adts
%% into \FO-formulas that is useful to compare $\foalt[\ell]$ with $\ADTk$. 
%
%
For the base cases of \adts $\epsilonadt$ and $\formula$, and we let: $\foformula_\epsilonadt\egdef \forall x \false$ and $\foformula_\formula\egdef \exists x (\forall y \lnot (x<y) \land \bigvee_{\val \models \formula} \val(x))$.

%% , where 
%% $\formulaFO(x)\egdef \bigvee_{\val \models \formula} \val(x)$ states
%% that ``propositional formula $\formula$ holds at position $x$ of the
%% trace''.

Now, regarding compound \adts, and not surprisingly, operator $\ORop$
is reflected by the logical disjunction: $\foformula_{\ORop(\att_1, \att_2)}\egdef \foformula_{\att_1} \lor \foformula_{\att_2}$, while operator $\Cop$ is
reflected by means of the logical conjunction with the negated second
argument: $\foformula_{\Cop(\att_1, \att_2)}\egdef \foformula_{\att_1} \land \lnot \foformula_{\att_2}$.
     On the contrary, the two remaining operators \SANDop and \ANDop
     require to split the trace into pieces, which can be captured by
     the folklore operation of \emph{left (resp.\ right) position
     relativizations} of \FO-formulas w.r.t.\ a
     position \cite[Proposition 2.1]{PERRIN1986393} (see also formulas
     of the form $\phi^{[x,y]}$
     in \cite{schiering1996counter}). Formally, given a position $x$
     in the trace $\trace$ and an \FO-formula $\foformula$, we define
     formula $\leftrelativ$ (resp.\ $\rightrelativ$) that holds of
     $\trace$ if the prefix (resp.\ suffix) of $\trace$ up to (resp.\
     from) position $x$ satisfies $\foformula$, as follows. For ${\bowtie} \in \{\leq,>\}$, we let:
 \begin{multicols}{2}

 $\somerelativ[\val(y)]=\val(y)$

 $\somerelativ[(y<z)]=(y<z)$

  $\somerelativ[\prop(y)]=\prop(y)$

  $\somerelativ[(\foformula \lor \foformula')]=\somerelativ[\foformula]\lor \somerelativ[\foformula']$ 

  $\somerelativ[(\lnot \foformula)]=\lnot\somerelativ$

  $\somerelativ[(\exists y \,\foformula)]=\exists y~(y\bowtie x \land \somerelativ)$

 %% $\leftrelativ[\val(y)]=\val(y)$

 %% $\leftrelativ[(y<z)]=(y<z)$

 %%  $\leftrelativ[\prop(y)]=\prop(y)$

 %%  $\leftrelativ[(\foformula \lor \foformula')]=\leftrelativ[\foformula]\lor \leftrelativ[\foformula']$ 

 %%  $\leftrelativ[(\lnot \foformula)]=\lnot\leftrelativ$

 %%  $\leftrelativ[(\exists y \,\foformula)]=\exists y~(y\leq x \land \leftrelativ)$

%%  $\rightrelativ[\val(y)]=\val(y)$

%%  $\rightrelativ[(y<z)]=(y<z)$

%%  $\rightrelativ[\prop(y)]=\prop(y)$

%%  $\rightrelativ[(\foformula \lor \foformula')]=\rightrelativ[\foformula]\lor \rightrelativ[\foformula']$

%%  $\rightrelativ[(\lnot \foformula)]=\lnot\rightrelativ$

%%  $\rightrelativ[(\exists y~\foformula)]=\exists y~ (y> x \land \rightrelativ)$
     \end{multicols}
    
Additionally, we write $\leftrelativ[\foformula][0]$ and
$\rightrelativ[\foformula][0]$ as the formulas obtained from
$\leftrelativ$ and $\rightrelativ$ by replacing every occurrence of
expressions $y\leq x$ and $y>x$ by $\false$ and $\true$ respectively.
%
%% \begin{example}
%% \spcom{add an example  with zero}
%% \end{example}
%
\begin{remark}
\label{foaltrelativization}
For every formula $\foformula \in \foalt[\ell]$, we also have $\leftrelativ[\foformula], \rightrelativ[\foformula] \in \foalt[\ell]$.
\end{remark}

We can now complete the translation from \ADT into \FO by letting
(w.l.o.g., by \Cref{rem:associative}, we can consider binary $\SANDop$ and $\ANDop$):
%\begin{itemize}

$\foformula_{\SANDop(\att_1, \att_2)}\egdef\exists x
    [\leftrelativ[\foformula_{\att_1}] \land \rightrelativ[\foformula_{\att_2}]] \lor
    (\leftrelativ[\foformula_{\att_1}][0] \land \foformula_{\att_2})$

    $\foformula_{\ANDop(\att_1, \att_2)}\egdef\exists x [(\leftrelativ[\foformula_{\att_1}] \land \foformula_{\att_2}) \lor (\leftrelativ[\foformula_{\att_2}] \land \foformula_{\att_1})] \lor
    (\leftrelativ[\foformula_{\att_1}][0] \land \foformula_{\att_2}) \lor
    (\leftrelativ[\foformula_{\att_2}][0] \land \foformula_{\att_1})$.\\
%% \begin{definition}
%% \label{def:ADTtoFO}
%% The first order formula $\foformula_\att$ associated with an \adt
%% $\att$ is defined inductively.
%% \begin{itemize}
%%     \item $\foformula_\epsilonadt= \forall x \false$;
%%     \item $\foformula_\formula= \exists x (\forall y \lnot (x<y) \land \hat{\formula}(x))$;
    %% \Item $\foformula_{\ORop(\att_1, \att_2)}= \foformula_{\att_1} \lor \foformula_{\att_2}$;
    %%  \item $\foformula_{\Cop(\att_1, \att_2)}= \foformula_{\att_1} \land \lnot \foformula_{\att_2}$;
 %%    \item $\foformula_{\SANDop(\att_1, \att_2)}=(\exists x (\leftrelativ[\foformula_{\att_1}] \land \rightrelativ[\foformula_{\att_2}])) \lor (\leftrelativ[\foformula_{\att_1}][0] \land \foformula_{\att_2})$;
%%     \item $\foformula_{\ANDop(\att_1, \att_2)}=(\exists x (\leftrelativ[\foformula_{\att_1}] \land \foformula_{\att_2} \lor \leftrelativ[\foformula_{\att_2}] \land \foformula_{\att_1} \lor (\leftrelativ[\foformula_{\att_1}][0] \land \foformula_{\att_2}) \lor (\leftrelativ[\foformula_{\att_2}][0] \land \foformula_{\att_1})$.
%% \end{itemize}
%% \end{definition}

%% \begin{remark}
%% \label{lem:sizefofromadt}
%% Remark that $\foformula_{\att}$ is of size exponential in $\sizeatt{\att}$.
%% \end{remark}
%@@\spcom{stopped here for the moment}
%%\Cref{lem:adttofo} below states the correctness of our translation $\att \mapsto \foformula_\att$ (see proof in \Cref{app:lem:adttofo}).
\begin{restatable}{lemma}{lemmaadttofo}
\label{lem:adttofo}\label{lem:sizefofromadt}
\begin{itemize}
\item For any $\trace \in \Traces$, $\trace \models \foformula_{\att}$ iff $\trace \in \semword{\att}$.
\item For any \adt $\att$, formula $\foformula_{\att}$ is of size exponential in $\sizeatt{\att}$. 
\end{itemize}
\end{restatable}

%% For the rest of the paper, for a set of trees $X$ and a set of
%% $\FO$-formulas $Y$, we say $X\subseteq Y$ whenever for each $\att X$,
%% there exists $\foformula\in Y$ such that, for each trace $\trace$, we
%% have $\trace\models \foformula$ iff $\trace \in \semword{\att}$.  we
%% say $Y\subseteq X$ whenever for each $\foformula\in Y$, there exists
%% $\att X$ such that, for each trace $\trace$, we have
%% $\trace\models \foformula$ iff $\trace \in \semword{\att}$.

In the rest of the paper, we use mere inclusion symbol $\subseteq$
between subclasses of \ADT and subclasses of \FO, with the canonical
meaning regarding the denoted trace languages.

An accurate inspection of the translation
$\att \mapsto \foformula_\att$ entails that every
$\ADTk$-definable \adt can be equivalently represented by a
$\foalt[k+2]$-formula, namely $\ADTk \subseteq \foalt[k+2]$.  However,
we significantly refine this expressiveness upperbound for $\ADTk$.

%\begin{lemma}
\begin{restatable}{lemma}{lemmaadtfoimproved}%{lemmaadtktofokplusun}%{lemmaadtsigmapi}
\label{lem:adtktofok+1}
\begin{enumerate}
\item\label{lem:adt0tofoalt2foalt2dual}  $\ADTk[0] \subseteq \foalt[2]\inter \foaltdual[2]$ -- with an effective translation.
\item\label{lem:itemadtktofok+1} For every $k > 0$, $\ADTk \subseteq \foalt[k+1]$ -- with an effective translation.
\end{enumerate}
\end{restatable}
%\end{lemma}

Regarding \Cref{lem:adt0tofoalt2foalt2dual} of \Cref{lem:adtktofok+1}, it can be observed that,
whenever $\att\in \ADTk[0]$, the quantifiers $\forall$ and $\exists$ commute in 
$\foformula_\att$ (see \Cref{app:lem:adt0tofoalt2foalt2dual}).
Now, for \Cref{lem:itemadtktofok+1}, the
proof is conducted by induction over $k$  (see \Cref{app:lem:adtktofok+1}). We sketch here the
case $k>1$.
 First, remark that if $\att_1,
..., \att_n \in \ADTk[k-1]$ are $\foalt[k]$-definable, then $\ORop(\att_1,...,\att_n), \SANDop(\att_1,...,\att_n)$ and $\ANDop(\att_1,...,\att_n)$ remain
$\foalt[k]$-definable, as formulas are obtained from conjunctions or disjunctions of $\foalt[k]$-definable formulas. Moreover,
if $\att_1$ and $\att_2$ are $\foalt[k]$-definable, then
$\att=\Cop(\att_1,\att_2)$ is 
$\foalt[k+1]$-definable as a formula can be obtained from a boolan combination of two $\foalt[k]$-definable formulas. Finally, with $k$ still fixed, it
can be shown by induction over the size of an \adt $\att$ that, if $\att \in \ADTk[k]$, since all its countermeasures
operators are of the form $\Cop(\att_1,\att_2)$ where $\att_1\in\ADTk$ and $\att_2\in \ADTk[k-1]$, we have
that \adt $\att$ is also $\foalt[k+1]$-definable, which concludes. \\
%\begin{lemma}
%% \begin{restatable}{lemma}{lemmaadtsigmapi}
%% \label{lem:adt0tofoalt2foalt2dual}
%%     $\ADTk[0] \subseteq \foalt[2]\inter \foaltdual[2]$, with an effective translation. 
%% \end{restatable}
%\end{lemma}

%See the proof in \Cref{app:lem:adtktofok+2}.
%We can do better than \Cref{lem:adtktofok+2} for $\ADTk[0]$.

We can also establish lowerbounds in the \ADT-hierarchy.%, through the levels $\dotdepth$ of the dot-depth hierarchy.
%can be embedded into the $\ADTk$ hierarchy.  We have that $\dotdepth{k}\subseteq\ADTk[2k+2]$.
\begin{restatable}{lemma}{lemmadotdepth}
\label{lem:dotdepth}
\begin{enumerate}
 \item\label{lem:withdotsdepth}   $\dotdepth\subseteq\ADTk[2\ell+2]$, and therefore $\foalt[\ell]\subseteq\ADTk[2\ell+2]$
%%\end{restatable}

%We can deduce as a direct corollary the result \cref{coro:fotoadtk}.

%%\begin{corollary}

%%\end{corollary}

%Remark that, for $\foalt[1]$, we have a more precise result.

%\begin{restatable}{lemma}{lemmafouninadtzero}
 \item\label{lem:fo1inADT0}
  $\foalt[1] \subseteq \ADTk[0]$ -- with an effective translation.
  \end{enumerate}
\end{restatable}

\Cref{lem:withdotsdepth} of \Cref{lem:dotdepth} is obtained by an induction of
$\ell$. Regarding \Cref{lem:fo1inADT0}, the
translation consists in putting the main
quantifier-free subformula of a $\foalt[1]$-formula in disjunctive
normal form, and to focus for each conjunct on the set of "ordering"
literals of the form $x<y$ or $\lnot (x<y)$ (leaving aside the other
literals of the form $\val(x)$ or $\lnot \val(x)$ for a while). Each
ordering literal naturally induces a partial order between the
variables. We expend this partial order constraint over the variables
as a disjunction of all its possible linearizations. For example, the
conjunct $x<y \land \lnot (y<z)$ is expended as the equivalent formula
$(x<y \land y=z) \lor (x<z \land z<y) \lor (z=x \land x<y) \lor
(z<x \land x<y)$.
%After this, associating an \adt to a conjunct is easy, for example, to $(x<y \land y=z)$, we associate $\SANDop(x,y\land z)$.
Now, each disjunct of this new formula, together with the constraints $\val(x)$ (or $\lnot \val(x)$), can easily be specified by a \SANDop-rooted \adt (ie. the root is a \SANDop). The initial $\foalt[1]$-formula then is associated with the \ORop-rooted tree that gathers all the aforementioned \SANDop-rooted subtrees (see \Cref{adt0shape} in \Cref{app:smallmodelproperty}).  Notice that the 
translation may induce at least an exponential
blow-up.

The reciprocal of \Cref{lem:fo1inADT0} in \Cref{lem:dotdepth} is an
open question. Still, we have little hope that it holds because
$\foformula_{\formula}$ seems to require a property expressible  
in $\foalt[2]\setminus \foalt[1]$.

\section{Strictness of the ADT-Hierarchy}
\label{sec-ADThierarchy}

One can notice that the \adt $\adtle[2] \in \ADTk[1]$  defines the finite language of
traces of length at most $1$, while it can be established that languages arising from \adts in
$\ADTk[0]$ are necessarily infinite -- if inhabited by a non-empty
word (see \Cref{remark-adt0sem} of \Cref{app:smallmodelproperty}). Thus, we can easily deduce $\ADTk[0] \subsetneq \ADTk[1]$. This section aims at showing that the entire \ADT-hierarchy is strict:

% \begin{proposition}
% \label{ADT1butnotADT0}
% $\ADTk[0] \subsetneq \ADTk[1]$.
% %% The finite language $\alphabettraces$ of
% %% traces of length at most $1$ is in $\ADTk[1]$ but not in $\ADTk[0]$.
% \end{proposition}
% Actually, the entire \ADT-hierarchy is strict. \tbcom{Pour gagner de la place, funsionner \Cref{ADT1butnotADT0} et \Cref{expressiveness}}

\begin{proposition}
\label{expressiveness}
For every $k\in \setn$, $\ADTk \subsetneq \ADTk[k+1]$, even if $\Prop=\{\prop\}$.
%For every $k\in \setn$, $\ADTk[k+1]$ is strictly more expressive than $\ADTk$.
%there exists a language that is $\ADTk[k+1]$-definable but not $\ADTk$-definable.
\end{proposition}

To show that $\ADTk\subsetneq \ADTk[k+1]$, we use a family of
languages, originally introduced in \cite{thomas1984application}, that
we write $\{\witness\}_{k \in \setn}$ over the two-letter alphabet
$\alphabet$ obtained from $\Prop=\{\prop\}$. For readability, we use
symbol $a$ (resp.\ $b$) for the valuation $\{p\}$ (resp.\ $\emptyset$)
of $\alphabet$. Formally, we define $\witness \subseteq \Traces$ as
follows -- where $\wolfgangmesure{\word}$ denotes the number of 
occurrences of $a$ minus the number of occurrences of  $b$ in the word $\word$:

We let $\word \in \witness$ whenever all the following holds.
\begin{itemize}
\item $\wolfgangmesure{\word}=0$;
\item for every $\word'\prefix\word$, $0 \leq \wolfgangmesure{\word'} \leq k$;
\item there exists $\word''\prefix\word$ s.t.\ $\wolfgangmesure{\word''}=k$.
\end{itemize}

In \cite[Theorem 2.1]{thomas1984application}, it is shown that
$\witness \in \dotdepth[k]\setminus \dotdepth[k-1]$, for all
$k\geq 1$.

% \begin{proposition}[{\cite[Theorem 2.1]{thomas1984application}}]
% \label{wthomas}
% $\witness \in \bool{\foalt}\setminus \bool{\foalt[k-1]}$, for all $k\geq 1$.                 
% \end{proposition}

We now determine the position of $\witness$ languages in the \ADT-hierarchy. 
We show \Cref{Ltoadt}.

\begin{restatable}{proposition}{lemmaltoadt}
\label{Ltoadt}
For each $k> 0$, $\witness \in \ADTk[k+1]\setminus\ADTk[k-2]$. 
\end{restatable}

First, $\witness\notin \ADTk[k-2]$ because   $\witness$ is
not $\foalt[k-1]$-definable (\cite{thomas1984application}). 
By an inductive argument over $k$ (see \Cref{app:Ltoadt}), we can build an \adt
$\attwitness \in \ADTk[k+1]$ that captures $\witness$. We only sketch here the case $k=1$.
For $\witness[1]$, we set
$\attwitness[1] \egdef \Cop(b,\ORop(\attall[\strict{b}], \allattall[\Cop(\adtge[2],\ANDop(\attall[a], \attall[b]))]))$,
depicted in \Cref{fig:attwitness1}. Note that
$\counterdepth(\attwitness[1])=2$ and that by a basic use of
semantics, we have $\semword{\attwitness[1]}=(ab)^+=\witness[1]$.\\

% \begin{subfigure}[b]{0.5\textwidth}         
% \centering
% \scalebox{0.8}{\input{forest-countermeasure}}
% \caption{A countermeasure from the Company.}
% \label{fig-adt-countermeasure}
% \end{subfigure}

\begin{figure}[h]
\begin{subfigure}[b]{0.47\textwidth}
\includegraphics[scale=.8]{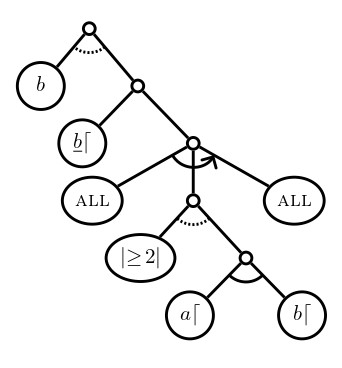}
    \caption{Representation of $\attwitness[1]$}
    \label{fig:attwitness1}
\end{subfigure}
\hfill
\begin{subfigure}[b]{0.53\textwidth}
    \scalebox{0.8}{\begin{tikzpicture}[xscale=0.5, yscale=0.46,node distance=1cm]

   \node (adt0) at (0,0) {$\ADTk[0]$};
   \node (sigma1) at (3, -1) {$\foalt[1]$};
   %\node (sigma0) [below=of sigma1] {$\foalt[0]$};
   \node (pass) [above=of sigma1] {};

   \node (adt1)  [above=of adt0] {$\ADTk[1]$};
   \node (adt2)  [above=of adt1] {$\ADTk[2]$};
   \node (adtdots)  [above=of adt2] {$\vdots$};

      \node (sigma2) [above=of pass] {$\foalt[2]$};
   \node (sigma3) [above=of sigma2] {$\foalt[3]$};
   \node (sigmadots) [above=of sigma3] {$\vdots$};
   
   \node (ab+) at (-2.5, 4.4) {$(ab)^+$};
   %\node (ab+) at (-2.8, 6.2) {$(ab)^+ \in$};
     \node(pi2) at (4.3, 4.4) {$\foaltdual[2]$};

      \path[-] 
    %(sigma0) edge node {} (sigma1)
    (sigma1) edge node {} (adt0)
    edge node {} (sigma2)
    (adt0) edge node {} (adt1)
    %edge node {} (sigma2)
    (sigma2) edge node {} (sigma3)
    (adt1) edge node {} (adt2)
    edge node {} (sigma2)
    (adt2) edge node {} (adtdots)
    edge node {} (sigma3)
    (sigma3) edge node {} (sigmadots)
    (adt0) edge node {} (pi2)
    (pi2) edge node {} (sigma1)
    edge node {} (sigma3)
         ;

    \path[->, dashed]

    (ab+) edge [bend right]  node [below=5pt, left= 3pt] {$\notin$} (adt1)
    edge [bend left] node [above=5pt, left=3pt] {$\in$}   (adt2);
   
    \node (adtkm1) at (8,1) {$\ADTk[k-1]$};
    \node(adtkm2)  [below= of adtkm1] {$\ADTk[k-2]$};
    \node (adtk) [above =of adtkm1] {$\ADTk$};
    \node (bk) at (12, 4) {$\dotdepth[k]$};
    \node (sigmak) at (12, 2.5) {$\foalt[k]$};
    \node (sigmakp1) at (12, 5.5) {$\foalt[k+1]$};
    \node (sigmakm1) at (12, -0.5) {$\foalt[k-1]$};
    \node (bkm1) at (12, 1) {$\dotdepth[k-1]$};
    \node (adtkp1)  [above= of adtk] {$\ADTk[k+1]$};
    \node (dots) at (8, 10) {$\vdots$};
    \node (adt2k) at (8, 14) {$\ADTk[2k+2]$};
    \node (bkp1) at (12, 7) {$\dotdepth[k+1]$};
    \node (sigmakp2)  at (12, 8.5) {$\foalt[k+2]$};

    \node (belowdots) at (8, -2.5) {$\vdots$};
    \node (abovedots) at (8, 18)  {$\vdots$};
    \node (belowdotsfoalt) at (12, -2) {$\vdots$};
    \node (abovedotsfoalt) at (12, 10){$\vdots$};

    \node (witness) at (6,2.5) {$\witness[k]$};

     \path[-] 
    %(sigma0) edge node {} (sigma1)
    
    (adtkm1) edge node {} (sigmak)
    edge node {} (adtk)
    (bk) edge node {} (adt2k)
    (sigmakm1) edge node {} (belowdotsfoalt)
    edge node {} (bkm1)
    (sigmak) edge node {} (bkm1)
    edge node {} (bk)
    (sigmakp1) edge node {} (bk)
    edge node {} (bkp1)
    (sigmakp2) edge node {} (bkp1)
    edge node {} (abovedotsfoalt)
    (adtk) edge node {} (adtkp1)
    edge node {} (sigmakp1)
    (adtkp1) edge node {} (dots)
    edge node {}(sigmakp2)
    (dots) edge node {} (adt2k)
    %(belowdots) edge node {} (adtkm2)
    (adt2k) edge node {} (abovedots)
    (adtkm2) edge node {} (adtkm1)
    edge node {} (sigmakm1)
    (bkm1) edge node {} (dots)
    (bkp1) edge node {} (abovedots);

    \path[->, dashed]

    (witness) edge [bend right]  node [below=5pt, left= 3pt] {$\notin$} (adtkm2)
    edge [bend left] node [above=5pt, left= 3pt] {$\in$}   (adtkp1);

\end{tikzpicture}}
    \caption{Summary of the results on \ADT-hierarchy}
    \label{diag-flevel}
\end{subfigure}
\end{figure}

% We now prove $\semword{\attwitness[1]}=(ab)^+=\witness[1]$.

% We point to \Cref{ex:adtsexamples} for the semantics for
% $\etrue, \coatt$ and $\allattall$.

% First of all, $\ANDop(\attall[a], \attall[b])$ defines
% the set of all traces with at least one occurrence of $a$ and one occurrence of
% $b$. Thus 
% %\[
% $\semword{\Cop(\adtge[2],\ANDop(\attall[a], \attall[b]))}=aa^+ \union
% bb^+$. Write $\att\egdef \Cop(\adtge[2],\ANDop(\attall[a], \attall[b]))$. 
% %\]
% Now, %We can then compute $\allattall[\att]$ 
% %\[
% $\semword{\allattall[\att]}=\Traces . (aa^+ \union bb^+).\Traces=\Traces . aa.\Traces \union \Traces.bb.\Traces.$ 
% %\]
% On this basis, \\
% $
% \semword{\attwitness[1]}=\semword{\Cop(b,\ORop(\attall[\strict{b}], \allattall[\att]))}=\Traces b \setminus (b\Traces\union \Traces . aa.\Traces \union \Traces.bb.\Traces)$ $=\witness[1]$. \tbcom{Discuter du niveau de détails/justifications à donner dans le corps du papier et dans l'annexe}

% With a similar method, we can define $\attwitness[1]^+$ and
% $\attwitness[1]^-$.

% In the remaining of the proof of \Cref{app:Ltoadt}, we construct
% by induction $\attwitness[k]$, $\attwitness[k]^+$ and
% $\attwitness[k]^-$.

% Now, if two consecutive levels of the $\ADT$ hierarchy are equals,
% then the whole hierarchy collapses after those levels (see the full
% proof in \cref{app:lem-collapse}):
% \begin{restatable}{lemma}{lemmacollapse}
%   \label{lem-collapse} If 
%  $\ADTk[k_0]=\ADTk[k_0+1]$ for some $k_0>0$, then all $\ADTk$ collapse from $k_0$.
% \end{restatable}

%\begin{proof}[\Cref{expressiveness}]

Now we have all the material to prove \Cref{expressiveness}. Indeed, assuming the hierarchy collapses at some level will contradict \Cref{Ltoadt}.
Notice that this argument is not constructive as we have no
witness of $\ADTk \subsetneq \ADTk[k+1]$ but  for $k=1$ where it can be shown that $(ab)^+ \in \ADTk[2]\setminus \ADTk[1]$ (see \Cref{app:smallmodelproperty}, \Cref{ex:ab+generator} and \Cref{prop:genadt1}).
% Indeed, it can be shown shown in \Cref{ex:ab+generator} \tbcom{attention, cet exemple est dans l'annexe} that there is no set of generators for $(ab)^+$. Therefore, from \Cref{prop:genadt1}, $(ab)^+$ is not $\ADTk[1]$-definable. \tbcom{Retravailler, on parle de generateurs, non défini dans le corps du papier.}
% The proof relies on \Cref{adt0languages} and consists in showing that
% any $\att \in \ADTk[1]$ with $(ab)^+ \subseteq \semword{\att}$ has a
% trace $\newtrace_0 \notin (ab)^+$ because of some factor of the form
% $(ab)^{2N}b$ where $N=\sizeatt{\att}$ (see the
% full proof in \Cref{app:ADT2butnotADT1}).
Our results about the \ADT-hierarchy are depicted on Figure~\ref{diag-flevel}.

% \begin{figure}[h]
%     \centering
%    \input{Diag-horizontal2}
%     \input{Diag-horizontalsmall}
%     \caption{Comparaison between the \ADT-hierarchy and the \FO alternation hierarchy}
%     \label{diag-flevel}
% \end{figure}
%\input{2023CONCUR/Diag-horizontal3}

% \begin{proof}
% Assume $\att\in \ADTk[2]$ a \AND-less \adt such that $\{\trace_1 (a(ab)^N b (ab)^N)^{2^{\sizeatt{\att}+1}} \trace_2| N \in \mathbb{N}\} \inter\semword{\att}$ is infinite. Then the expression $\trace_N=\trace_1 a(ab)^N b (ab)^N)^{2^{\sizeatt{\att}+1}} \trace_2 \in \semword{\att}$ is true for a infinite choice of $N\in \mathbb{N}$.
% \end{proof}

\section{Decision Problems on Attack-Defense Trees}
\label{sec-decisionproblems}

We study classical decision problems on languages, through the
lens of \adts, with a focus on the role played by the countermeasure-depth in
their complexities. The problems are the following.

%\begin{definition}~
\begin{itemize}
\item The \emph{membership} problem, written \ADTmemb, is defined by:\\
    \ADTmembpb
\item The \emph{non-emptiness} problem, written \ADTNE, is defined by:\\
  \ADTNEpb
%% \item The \emph{equivalence} problem, written \ADTequiv, is is defined by:\\
%%   \ADTequivpb
  \end{itemize}
%\end{definition}

We use notations \ADTmemb[k] and \ADTNE[k] whenever the input \adts of
the respective decision problems are in $\ADTk$, with a fixed $k$. Our
results are summarized in \Cref{table:results}, and we below comment on them, row by row.

\begin{table}[h]
\begin{threeparttable}[b]
%\caption{...}
%\begin{center}
\begin{tabular}{|c|c|c|c|c|}
%\hline%
\cline{2-5}
%\backslashbox{Problem}{Inputs}
\multicolumn{1}{c|}{}
& $\ADTk[0]$ & $\ADTk[1]$ & $\ADTk$ $(k\geq 2)$ & $\ADT$ \\
 \hline
 \ADTmemb & \PTIME     & \PTIME     & \PTIME  & \PTIME \\
         %&   (Prop.\ \ref{mem:peasy})  & (Prop.\ \ref{mem:peasy}) & (Prop.\ \ref{mem:peasy})  & (Prop.\ \ref{mem:peasy}) \\
 \hline
 \ADTNE & \NP-comp & \NP-comp & $(k+1)$-\EXPSPACE & non-elem \\
          %& (Prop.\ \ref{ADTNE0andADTNE1NPc}) & (Prop.\ \ref{ADTNE0andADTNE1NPc}) & (Prop.\ \ref{ADTNEkupperbound})  & (Prop.\ \ref{ne:nonelementary}) \\
        &              &                & \tnote{1}~~$\geq$ $\NSPACE(g(k-5,c \sqrt{\frac{n-1}{3}}))$           &                \\
                % &              &              & (Prop.\ \ref{ADTNEklowerbound})  &  \\
 \hline
 \ADTequiv & \coNP-comp & $4$-\EXPSPACE   & $(k+2)$-\EXPSPACE & non-elem \\
           %& (Prop.\ \ref{allADTequiv}) & (Prop.\ \ref{allADTequiv}) & (Prop.\ \ref{allADTequiv}) & (Prop.\ \ref{allADTequiv}) \\
        &              &           &  \tnote{2}~~$\geq$ $\NSPACE(g(k-4,c \sqrt{\frac{n-1}{3}}))$             &                \\
             % &              &             & (Prop.\ \ref{allADTequiv})  &  \\
 \hline
\end{tabular}
%\end{center}
\begin{tablenotes}
\item [1] if $k \geq 5$.
\item [2] if $k \geq 4$.
\end{tablenotes}
\end{threeparttable}
\caption{Computational complexities of decision problems on \adts.\label{table:results}}
\end{table}

Regarding \ADTmemb (first row of \Cref{table:results}), we recall that
\adts and \shseres are expressively equivalent, but with a
translation (\Cref{theo:starfreeequalsADT}) from the former to the
latter that is not polynomial. We therefore cannot exploit \cite[Theorem
  2]{kupferman2002improved} for a \PTIME complexity of the word
membership problem for \shseres, and have instead developed a
dedicated alternating logarithmic-space algorithm in
\Cref{app:sec-decisionproblems}.

Regarding \ADTNE (second row of \Cref{table:results}), and because
\shseres can be translated  as \adts  (see ``if'' direction in the proof
of \Cref{theo:starfreeequalsADT}), the problem \ADTNE inherits from
the hardness of the non-emptiness of \shseres
\cite[p.\ 162]{stockmeyer1974complexity}. In its full generality, \ADTNE is therefore
non-elementary (last column). Moreover, by our exponential translation
of $\ADTk$ into the \FO-fragment $\foalt[k+1]$
(\Cref{lem:adtktofok+1}), we obtain the $(k+1)$-\EXPSPACE upper-bound
complexity for \ADTNE[k] (recall satisfiability problem for
$\foalt[\ell]$ is $(\ell-1)$-\EXPSPACE). Additionally, a lower-bound
for \ADTNE[k] is directly given by \cite[Theorem
  4.29]{stockmeyer1974complexity}: for \shseres that linearly
translate as \adts with countermeasure-depth $k$, their non-emptiness
is at least $\NSPACE(g(k-5,c \sqrt{\frac{n-1}{3}}))$. Interestingly,
the Small Model Property (\Cref{boundedsizetracesADT1}) yields an \NP
upper-bound complexity for \ADTNE[1], also applicable for \ADTNE[0], that is optimal since one can
reduce the \NP-complete satisfiability of propositional formulas to
the non-emptiness of leaf \adts. 

Finally regarding \ADTequiv (last row of \Cref{table:results}), one
can observe that \ADTNE and \ADTequiv are very close.  First, \ADTNE
is a particular case of \ADTequiv with the second input \adt
$\att_2\egdef \emptyset$.  As a consequence, \ADTequiv inherits from
the hardness of \ADTNE, and is therefore non-elementary (last column).
Also, because deciding the equivalence between $\att_1$ and $\att_2$
amounts to deciding whether $\ORop (\Cop(\att_1, \att_2), \Cop(\att_2,
\att_1))=\emptyset$, we get reduction from \ADTequiv[k] into
\ADTNE[k+1] which provides the results announced in
Columns 1-3.

\section{Discussion}
\label{sec-conclusion}

\hspace{\parindent} First, we discuss our  model of \adts with regard to the literature.
However, we do not compare with settings where \adts leaves are actions \cite{mauw2005foundations}, as they yield only finite languages, and address other issues \cite{widel2019beyond}.

\noindent Our \adts have particular features, but remain somehow standard. Regarding the syntax, firstly, even though we did not type our nodes as
proponent/opponent, the countermeasure operator fully determines the alternation between attack and defense. Also, we introduced the non-standard leaf $\emptyleaf$ for the singleton empty-trace language, 
not considered in the literature. Still, it is a very natural object in the formal language landscape,
and anyhow does not impact our overall computational complexity analysis.
Regarding the semantics, it can be shown that the definition of our
operator \ANDop together with the reachability goal semantics of the leaves coincides
with the acknowledged semantics considered in \cite{audinot2017my,brihaye2022adversarialgandalf}.

Second, we  discuss our results. 
We showed the strictness of the $\ADT$-hierarchy in a non-constructive manner. However, exhibiting an element of
$\ADTk[k+1]\setminus \ADTk$ ($k\geq 2$) is still an open question. Since
$\witness[1] \in \ADTk[2]\setminus\ADTk[1]$, languages $\witness$ are natural candidates. We conjecture it is the case and we are currently working on 
Erenfeucht-Fraissé(EF)-like games for \adts (in the spirit
of \cite{thomas1984application} for \shseres) to prove it. Moreover, EF-like games for \adts  may also help to better compare the \ADT-hierarchy and the \FO alternation hierarchy, in particular, whether the hierarchies eventually coincide. For now, finding a tighter inclusion of
$\foalt[\ell]$ in some $\ADTk$ for each $\ell$ seems difficult; recall that we established $\foalt[\ell]\subseteq \ADTk[2\ell+2]$. Any progress in this line would be of great help to obtain tight complexity bounds for \ADTNE[k] and \ADTequiv[k]. 
%
%% faux \FO2 in NEXPTIME
%% By \Cref{ADTNE0andADTNE1NPc}, we may wonder if $\ADTk[0]$ and
%% $\ADTk[1]$ can be expressed in $\FO^2$, as they share the same
%% complexity for the non-emptiness problem. However our generic
%% translation from $\ADTk$ to \FO-formula in \Cref{lem:adttofo} does not
%% allow us to conclude anything.
%
 Finally, determining the level of a language in the \ADT-hierarchy seems as hard as determining its level in the dot-depth hierarchy, recognised as a difficult question
\cite{place2015tale,DBLP:conf/birthday/DiekertG08}.

\bibliography{ref}

\appendix
\section{Figure of $\att_{ex_2}$ in \Cref{ex:introductory}} \label{app:fig-countermeasureofcountermeasure}

\begin{figure}[H]
    \centering
\includegraphics[scale=1]{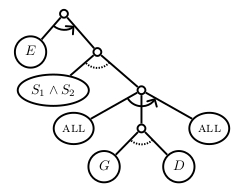}
\end{figure}

% \section{Proof of \Cref{remark-adt0sem}}
% %----------------------------------
%

\section{Proof of \Cref{boundedsizetracesADT1}}%\Cref{sec-smallmodelproperty}}
\label{app:smallmodelproperty}

We consider the \emph{lift}
binary relation between traces (\Cref{def:lift}), a slight variant of
the classic relation of super-word in language theory. We use the lift
to define generators of a set of traces (see \Cref{def:generator}) for
$\att$ an \adt. We define (in \Cref{def:gen}) $\gen{\att}$, a finite
set of traces, and we show in \Cref{prop:genadt1} that, if
$\att \in \ADTk[1]$, then $\gen{\att}$ is a set of generators of the
semantics of $\att$. We also show in \Cref{prop:genbounded} that an
element of $\gen{\att}$ has its size bounded by the number of leaves
of $\att$. \Cref{prop:genbounded} and \Cref{prop:genadt1} are enough
to deduce the small model property of $\ADTk[1]$, stated
in \Cref{boundedsizetracesADT1}. 

\begin{definition}
\label{def:lift}
A trace $\trace$ is a $\emph{lift}$ of a trace
$\trace'=\val_1...\val_n$, written $\trace' \lift \trace$, whenever 
$\trace\in\Traces \val_1... \Traces\val_n$. 
\end{definition}

Notice that $\lift$ is an order over $\Traces$ and that lift relation is like the relation of super-word where the last letters of the respective traces are
identical. It can be shown by induction of trees that every $\att \in \ADTk[0]$, the set
$\semword{\att}$ is $\lift$-upward closed, and is therefore infinite:

\begin{lemma}\label{remark-adt0sem}
Let $\att \in \ADTk[0]$. If $\trace \in \semword{\att}$, then for each trace $\trace'$ such that $\trace\lift\trace'$, we have $\trace'\in \semword{\att}$.  
\end{lemma}

 \begin{proof}
%\label{remark-adt0sem}
 We show this result by induction over the shape of an \adt.
 
 If $\att=\emptyleaf$, then there is no lift of $\eword$. If $\att$ is a non-empty leaf $\formula$, the result holds since the lift relation preserves the last letter. 
 
  Now, we assume that the result holds for two \adts $\att_1$ and $\att_2$ and we need to show the result for $\ORop(\att_1, \att_2)$, $\SANDop(\att_1, \att_2)$ and $\ANDop(\att_1, \att_2)$.
  
  The three cases are direct and similar, we only write the case $\att=\SANDop(\att_1, \att_2)$:
    %   $\att=\ORop(\att_1, \att_2)$. If $\trace\in \semword{\att}$ and $\trace\lift\trace'$ for a certain trace $\trace'$, we have either $\trace \in \semword{\att_1}$ or $\trace \in \semword{\att_1}$ and we can conclude by our induction hypothesis.
  if $\trace\in \semword{\att}$ and $\trace\lift\trace'$ for a certain trace $\trace'$, then we can write $\trace=\trace_1\concat\trace_2$ with $\trace_1\in \semword{\att_1}$ and $\trace_2 \in \semword{\att_2}$. From the structure of a lift, we can also write $\trace'=\trace'_1\concat\trace'_2$ such that $\trace_1\lift\trace'_1$ and $\trace_2\lift\trace'_2$, we can then conclude with our induction hypothesis. 
 \end{proof}

From \Cref{remark-adt0sem}, it is immediate that any non-empty
 semantics of an \adt in $\ADTk[0]$ is necessarily infinite. We now turn to the notion of
 \emph{generators} to establish the Small Model Property.

\begin{definition}
\label{def:generator}
Let $\traceset \subseteq \Traces$. A \emph{set of generators of
$\traceset$} is a finite subset $G$ of $\traceset$ with the following
property: for each non-empty trace $\trace_\traceset \in \traceset$ there exists
$\trace_{gen}=\val_1...\val_n\in G$ such that we can write $\trace_\traceset=\word_1\val_1...\word_n\val_n$ and for each $\trace \in \Traces$, if we can write $\trace=\word_1'\val_1...\word_1'\val_n$ with $\word_1'\val_1\lift\word_1\val_1$, ..., $\word_n'\val_n\lift\word_n\val_n$, then $\trace\in \traceset$.
\end{definition}

\begin{example}
For the singleton set of propositions $\{\prop\}$, we consider the set of traces
$\traceset=\prop^+\setminus\{\prop\prop\}$. The set
$\{\prop, \prop\prop\prop\}$ is a set of generators for $T$. However,
the set $\{\prop\}$ is not because we have $\prop\prop\prop\in T$ and
$\prop\lift \prop\prop\prop$ but we also have
$\prop \lift \prop\prop\lift\prop\prop\prop$ and $\prop\prop\notin T$.
\end{example}

\begin{example}
\label{ex:ab+generator}
For the singleton set of propositions $\{\prop\}$, we let $a$ be the valuation
that makes $\prop$ true and $b$ the valuation that makes $\prop$ false. Thus each trace is a
word over the alphabet $\{a,b\}$. Now, remark that the set $(ab)^+$ has no set of
generators. Indeed, assume that $G \subseteq (ab)^+$ is a set of
generators. We write $(ab)^N$ for the largest element of $G$, thus
$(ab)^Nab \in (ab)^+ \setminus G$. However, for each $\trace_{gen}\in
G$, we have $\trace_{gen}\lift (ab)^Nb \lift (ab)^Nab$ and
$(ab)^Nb\notin (ab)^+$, which contradicts \Cref{def:generator}.
\end{example}

We now define the set $gen(\att)$ for each \adt $\att$ that will be our
candidate to show that every
of \adts in $\ADTk[1]$ has a set of generators.

\begin{definition}
\label{def:gen}
Let $\att \in \ADTk[1]$. The set $\gen{\att}$ is inductively defined as follows: 
\begin{itemize}
    \item $\gen{\emptyleaf}=\emptyset$ and $\gen{\formula}= \{\val \in 2^\Prop : \val \models \formula \}$,
    \item $\gen{\ORop(\att_1, \att_2)}= \gen{\att_1}\union \gen{\att_2}$,
    \item $\gen{\SANDop(\att_1, \att_2)}= \gen{\att_1}\concat \gen{\att_2}$,
    \item  $\gen{\Cop(\att_1, \att_2)}= \gen{\att_1}\setminus \semword{\att_2}$,
    \item $\gen{\ANDop(\att_1, \att_2)}= (\gen{\att_1}\shuff \gen{\att_2})\inter \semword{\ANDop(\att_1, \att_2)}$, where $\shuff$ is the classic shuffle operator inductively defined as:
   $\eword \shuff \trace=\trace \shuff \eword=\{\trace\}$,
     $\val_1\trace_1 \shuff \val_2\trace_2= \val_1 \concat (\trace_1 \shuff \val_2\trace_2) \union \val_2 \concat (\val_1\trace_1 \shuff \trace_2)$ and
     $\val\trace_1 \shuff \val\trace_2= \val \concat (\trace_1 \shuff \val\trace_2) \union \val \concat (\val\trace_1 \shuff \trace_2) \union \val \concat (\trace_1 \shuff \trace_2)$.
\end{itemize}

\end{definition}

 We can prove the following \Cref{prop:genbounded} by induction over \adts.

\begin{restatable}{lemma}{lemmagenbounded}
\label{prop:genbounded}
For each trace $\trace \in \gen{\att}$, its size $\sizetrace{\trace}$ is bounded by the number of leaves of $\att$. 
\end{restatable}

\begin{proof}
    
 We show this result by induction over the shape of an \adt.
 
 If the \adt is the empty leaf $\emptyleaf$, then there is no element in $\gen{\att}$. If \adt is a non-empty leaf $\formula$ then all elements of $\gen{\att}$ are of size 1. 
 
 Now, we assume that the result holds for two \adts $\att_1$ and $\att_2$ and we show the result for $\gen{\ORop(\att_1, \att_2)}$, $\gen{\SANDop(\att_1, \att_2)}$, $\gen{\ANDop(\att_1, \att_2)}$ and $\gen{\Cop(\att_1,\att_2)}$. For the rest of the proof, we write $n_1$ the number of leaves of $\att_1$ and $n_2$ the number of leaves of $\att_2$.
 
 By induction, if $\trace \in \gen{\ORop(\att_1, \att_2)}$, then $\sizetrace{\trace}\leq MAX\{n_1, n_2\}$. If $\trace\in \gen{\SANDop(\att_1, \att_2)}$ or $\trace\in \gen{\ANDop(\att_1, \att_2)}$, then $\sizetrace{\trace}\leq n_1+n_2$. Finally, if $\trace\in \gen{\Cop(\att_1, \att_2)}$, then $\sizetrace{\trace}\leq n_1$. In each case, $\sizetrace{\trace}\leq n_1+n_2$ and $n_1+n_2$ is the number of leaves of the four studied \adts, which concludes. 
\end{proof}  

If $\att$ has no nested countermeasures yields the following. 

\begin{restatable}{proposition}{propgenadt}
\label{prop:genadt1}
When $\att \in \ADTk[1]$, the set $\gen{\att}$ is a set of generators for $\semword{\att}$.
\end{restatable}

\begin{proof}
 
We show this result by induction over the shape of an \adt.

If $\att=\emptyleaf$, then its semantics contains only the empty trace. Thus $\emptyset$ is a set of generators. If $\att=\formula$, then $\semword{\formula}=\{\val_1...\val_n\in \Traces:\val_n \models\formula\}=\Traces\concat \gen{\formula}$, thus $\gen{\formula}\subseteq \semword{\formula}$ and is a set of generators.

Now, for two \adts $\att_1$ and $\att_2$, we assume that $\gen{\att_1}$ and $\gen{\att_2}$ are sets of generators and we show that $\gen{\ORop(\att_1, \att_2)}$, $\gen{\SANDop(\att_1, \att_2)}$ and $\gen{\ANDop(\att_1, \att_2)}$ are sets of generators for their corresponding \adt.

First of all, for $\att \in \{\ORop(\att_1, \att_2), \SANDop(\att_1, \att_2), \ANDop(\att_1, \att_2), \Cop(\att_1, \att_2)\}$, we have $\gen{\att}\subseteq \semword{\att}$. Indeed,
for $ \gen{\ORop(\att_1, \att_2)}$ and $ \gen{\SANDop(\att_1, \att_2)}$, the composition rule is the same as the one used for $\semword{\att}$. Moreover, for $ \gen{\ANDop(\att_1, \att_2)}$ and $ \gen{\Cop(\att_1, \att_2)}$, we intersect with $\semword{\att}$. Furthermore, from \cref{prop:genbounded}, each element of $\gen{\att}$ is bounded, thus $\gen{\att}$ is a finite set. 

It remains to show that for each $\trace_\att \in \semword{\att}$, we have that there exists $\trace_{gen}=\val_1...\val_n\in \gen{\att}$ such that, $\trace_\att=\trace_1\val_1...\trace_n\val_n$ and for each trace $\trace=\trace_1'\val_1...\trace_n'\val_n$ with $\trace_1'\val_1\lift\trace_1\val_1, ..., \trace_n'\val_n\lift\trace_n\val_n$, we have $\trace\in \semword{\att}$. We distinguish the four \adts operators:

\begin{itemize}
    \item $\att=\ORop(\att_1, \att_2)$. If $\trace_\att \in \semword{\att}$, then either $\trace_\att\in \semword{\att_1}$ or $\trace_\att\in \semword{\att_2}$. With no loss of generality, we assume  $\trace_\att\in \semword{\att_1}$. Since $\gen{\att_1}$ is a set of generators for $\att_1$, there exists $\trace_{gen}=\val_1...\val_n \in \gen{\att_1} \subseteq \gen{\att}$ such that $\trace_\att=\trace_1\val_1...\trace_n\val_n$ and for each trace $\trace=\trace_1'\val_1...\trace_n'\val_n$ with $\trace_1'\val_1\lift\trace_1\val_1, ..., \trace_n'\val_n\lift\trace_n\val_n$, we have $\trace\in \semword{\att_1} \subseteq \gen{\att}$.
    
    \item $\att=\SANDop(\att_1, \att_2)$. If $\trace_\att \in \semword{\att}$, there exists $\trace_{gen}=\val_1...\val_n\in \gen{\att}$ such that $\trace_\att=\trace_1\val_1...\trace_n\val_n$. Indeed, we can write $\trace_\att=\trace_1\concat \trace_2$ with $\trace_1\in \semword{\att_1}$ and $\trace_2\in \semword{\att_2}$, thus there exists $\trace_{gen_1}=\val_1...\val_i$ and $\trace_{gen_2}=\val_{i+1}...\val_n$ with $\trace_{gen_1} \in \gen{\att_1}$ and $\trace_{gen_2} \in \gen{\att_2}$. Therefore we can consider $\trace_{gen}=\trace_{gen_1}\concat\trace_{gen_2}$. Moreover, the following holds: property for each trace $\trace=\trace_1'\val_1...\trace_n'\val_n$ with $\trace_1'\val_1\lift\trace_1\val_1, ..., \trace_n'\val_n\lift\trace_n\val_n$, we have $\trace\in \semword{\att_1} \subseteq \gen{\att}$. Indeed we can simply write $\trace=\trace_a\concat\trace_b$ with $\trace_a=\trace'_1 \val_1... \trace_i'\val_i$ and $\trace_b=\trace_{i+1}'\val_{i+1}...\trace_n\val_n$. So, $\trace_a\in \semword{\att_1}$ and $\trace_b\in \semword{\att_2}$ and we conclude $\trace\in \semword{\att}$.
    
    \item $\att=\ANDop(\att_1, \att_2)$. First, let's remark that, for a trace $\trace_\att\in\semword{\att}$, there exists $\trace_1 \in \semword{\att_1}$ and $\trace_2 \in \semword{\att_2}$ and $\trace_{gen}=\val_1...\val_n\in (\trace_1\shuff\trace_2) $ such that $\trace_\att=\trace_1\val_1...\trace_n\val_n$. Indeed, we can write $\trace_\att=\trace_a\concat\trace_b$ with either $\trace_a \in \semword{\att_1}$ and $\trace_\att \in \semword{\att_2}$ or $\trace_a \in \semword{\att_2}$ and $\trace_\att \in \semword{\att_1}$. If we assume the former case (the latter is symmetrical), then there exists $\trace_{gen_1}=b_1...b_m\in \gen{\att_1}$ such that $\trace_a=u_1b_1...u_mb_m$ and there exists $\trace_{gen_2}=c_1...c_p\in \gen{\att_2}$ such that $\trace_\att=v_1c_1...v_pc_p$. If we consider the trace $\trace_{gen}$ as the trace obtained from $\trace_\att$ by removing all valuations appearing neither in $\trace_{gen_1}$ nor in $\trace_{gen_2}$. Then, by construction: $\trace_{gen}=\val_1...\val_n\in (\trace_{gen_1}\shuff\trace_{gen_2})$ and $\trace_\att=\trace_1\val_1...\trace_n\val_n$. The rest of the proof is similar to the $\SANDop$ as we show that a trace $\trace=\trace_1'\val_1...\trace_n'\val_n$ with $\trace_1'\val_1\lift\trace_1\val_1, ..., \trace_n'\val_n\lift\trace_n\val_n$ always can be written as $\trace=\trace_1\concat\trace_2$ with $\trace_1\in \semword{\att_1}$ and $\trace\in \semword{\att_2}$. 
    
    \item $\att=\Cop(\att_1, \att_2)$. For a trace $\trace_\att\in \semword{\att}$, there exists $\trace_{gen_1}=\val_1...\val_n\in \gen{\att_1}$ such that, $\trace=\trace_1'\val_1...\trace_n'\val_n$ with $\trace_1'\val_1\lift\trace_1\val_1, ..., \trace_n'\val_n\lift\trace_n\val_n$, we have $\trace\in \semword{\att_1}$. Moreover, if $\trace_\att \in \semword{\Cop(\att_1,\att_2)}$, then $\trace_\att \notin \semword{\att_2}$. Furthermore, if $\Cop(\att_1, \att_2)$ is in $\ADTk[1]$, then $\att_2\in \ADTk[0]$, thus from \Cref{remark-adt0sem}, for each $\trace'$ such that $\trace' \lift\trace_\att$ we have $\trace'\notin\semword{\att_2}$. This shows us that $\trace_{gen_1}\in \semword{\Cop(\att_1,\att_2)}$ and by taking $\trace_{gen}=\trace_{gen_1}$, the property we need to prove holds. 
\end{itemize}
\end{proof}

\begin{remark}
\label{adt0shape}
    An \adt $\att\in \ADTk[0]$ with a non-empty semantics always can be written as an $\ORop$ over $\SANDop$ over leaves.
Indeed, for an \adt $\att\in \ADTk[0]$ and its set of generator $\gen{\att}=\{\trace_1, ...,\trace_n\}$, by combining \Cref{prop:genadt1} and \Cref{remark-adt0sem}, we have: $\semword{\att}=\underset{\val_1...\val_m \in \gen{\att}}{\bigcup} \Traces\val_1...\Traces\val_m$. Therefore, for  we can write $\att=\ORop(\att_{\trace_1}, ..., \att_{\trace_n})$, Where $\att_{\val_1...\val_m}=\SANDop(\formula_1, ..., \formula_m)$ with $\formula_j$ is the formula only satisfied by $\val_j$. 
\end{remark}

By \Cref{prop:genadt1},  for $\att\in\ADTk[1]$,  if some non-empty trace $\trace_\att \in \semword{\att}$, there is
some
$\trace_{gen}\in \gen{\att}$ with
$\trace_{gen}\lift \trace_\att$ that is a witness. Together with \Cref{prop:genbounded},
we conclude the proof of \Cref{boundedsizetracesADT1}.
Remark that, since
$\ADTk[0]\subseteq\ADTk[1]$, the result also holds for $\ADTk[0]$.

Observe that our proof technique heavily relies on the notion of set of
generators, which calls for the existence of a set of generators for
\adts above $\ADTk[1]$.  There is an \adt in $\ADTk[2]$ that shows it
is hopeless: this \adt specifies the language $(ab)^+$ (see
\Cref{sec-ADThierarchy}) but has no set of generators (see
\Cref{ex:ab+generator}).\\

% \section{Proof of \Cref{prop:genbounded}}
%----------------------------------
  
% \section{Proof of \Cref{prop:genadt1}}
% %----------------------------------

%\propgenadt*

%% \atcom{Notes pour Sophie:}

%% Why can't we use the following definition ? 

%% \begin{definition}
%% \label{def:generator}
%% Let $\traceset \subseteq \Traces$. A \emph{set of generators of
%% $\traceset$} is a finite subset $G$ of $\traceset$ with the following
%% properties: (a) each non-empty trace $\trace_\traceset \in \traceset$ has an element of $G$ $\lift$-below (\ie some
%% $\trace_{gen}\in G$ with $\trace_{gen}\lift \trace_\traceset$), and
%% (b) $G$ is $\lift$-convexe (\ie for each $\trace_{gen} \in G$,
%% $\trace_{\traceset} \in \traceset$, and $\trace \in \Traces$,
%% $\trace_{gen}\lift\trace\lift\trace_{\traceset}$ entails
%% $\trace\in \traceset$).
%% \end{definition}

%% Assume the following \adt: $\att=\SANDop(\Cop(\SANDop(a,b),\attall{d}), c)$, we have $\semword{\att}=((\Traces a \Traces b)\setminus (\Traces d \Traces))\concat\Traces c$. Therefore, we want as set of generator $\{abc\}$ but $abc \lift adbc \lift abdbc$. However, $adbc \not \in \semword{\att}$ and $abdbc \in \semword{\att}$.

%% Moreover, we cannot have a iff for the set of generator result because the set $\{\trace | aaa \lift \trace \}\setminus\{\trace | ababa\lift \trace\}$ accept a set of generator but I don't know how to express it using \adts in $\ADTk[1]$. (I conjecture it is not possible)

\section{Complements of \Cref{sec-expressiveness}}
%----------------------------------
\label{app:lem:adt0tofoalt2foalt2dual}

\label{app:lem:adttofo}\lemmaadttofo*

\begin{proof}
We will show this result by induction over the \adt.

If $\att=\epsilon$, then we have clearly that $\foformula_\epsilon=\{\eword\}$. Similarly, for a Boolean formula $\formula$, we have that $\Models{\foformula_\formula}$ is all traces finishing with a valuation $\val$ such that $\val \models \formula$. Thus, the property holds for leaves attack-defense trees.

Now, if we assume that the property is true for $\att_1$ and $\att_2$, then, since semantics for the $\ORop$ operator is a union of sets and semantics of $C$ operator is a difference of sets, the property for $\foformula_{\ORop(\att_1, \att_2)}$ and $\foformula_{\Cop(\att_1, \att_2)}$ trivially holds.

The formulas for $\ANDop$ and $\SANDop$ operators are a little more difficult to understand since we need to distinguish whether an empty trace is in semantics of $\att_1$ or $\att_2$. Indeed, a $\trace$ is in semantics of $\SANDop(\att_1,\att_2)$ if and only if we can find $\trace_1$ and $\trace_2$ such that $\trace_1 \in \semword{\att_1}$ and $\trace_2 \in \semword{\att_2}$ if $\trace_1$ is nonempty, then this is equivalent to say that there exists a position in the trace where we can cut as done in the first parenthesis of $\foformula_{\SANDop(\att_1, \att_2)}$. Otherwise, if $\trace_1$ is empty, then we handle this case with the use of $\foformulalow{0}_{\att_1}$ in the second parenthesis of $\foformula_{\SANDop(\att_1, \att_2)}$. The idea for the $\ANDop$ operator is completely symmetrical except that we also need to consider the case where $\att_2$ has the empty trace in its semantics.

Finally, it is clear that $\foformula_\att$ is of size exponential as subformulas for subtrees may be duplicated (for $\SANDop$ and $\ANDop$).

\end{proof}

For the rest of this section, we often switch from an \adt $\att$ to its corresponding \FO-formulas $\foformula_\att$ as defined in \Cref{sec-expressiveness}, always while assuming \Cref{lem:adttofo}.

Let us recall some basic results on $\foalt$ and $\foaltdual$ hierarchies.

\begin{remark}[\cite{thomas1982classifying}, Lemma 2.4]
\label{easyfoalt}
    \begin{itemize}
        \item[(a)] The negation of a $\foalt$-formula is equivalent to a $\foaltdual$-formula.
        \item[(b)] A disjunction or conjunction of $\foalt$-formulas is equivalent to a $\foalt$ formula.
        \item[(c)] A Boolean combination of $\foalt$-formula is equivalent to a $\foalt[k+1]$-formula.
        \item[(d)] The statements (a)-(c) hold in dual form for $\foaltdual$-formulas.
    \end{itemize}
\end{remark}

\label{app:lem:adtktofok+1}\lemmaadtfoimproved*
%\lemmaadtsigmapi*

%We now start the proof of \Cref{lem:adt0tofoalt2foalt2dual}

\begin{proof}
Regarding \Cref{lem:adt0tofoalt2foalt2dual}, for both cases $\ADTk[0]
\subseteq \foalt[2]$ and $\ADTk[0] \subseteq \foaltdual[2]$, we
construct for each $\att \in \ADTk[0]$ an $\FO$-formula
$\foformula_\att$ by induction over $\att$ such that
$\semword{\att}=\{\trace \in \Traces: \trace \models
\foformula_\att\}$.

To prove that $\ADTk[0] \subseteq \foalt[2]$, one can refer to the
translation from \ADT into \FO, and easily verify that formula $\foformula_\att \in \foalt[2]$.

%% by induction over \adts.  For leaf
%% \adts, we have:
%% \begin{itemize}
%%     \item $\foformula_\epsilonadt\egdef \forall x \false$
%%     \item $\foformula_{\formula}= \exists x \forall y ( \lnot (x<y) \land \hat{\formula}(x))$
%% \end{itemize}
%% It follows $\foformula_\epsilonadt$ and $\foformula_{\formula}\in \foalt[2]$.

%% For a general \adt $\att \in \ADTk[0]$, $\att$ do not contain any $\Cop$ operators. Therefore, $\foformula_\att$ is inductively constructed on the basis of the following formulas: 
%% \begin{itemize}
%%    \item $\foformula_{\ORop(\att_1, \att_2)}\egdef \foformula_{\att_1} \lor \foformula_{\att_2}$;
%%    \item $\foformula_{\SANDop(\att_1, \att_2)}=(\exists x
%%     (\leftrelativ[\foformula_{\att_1}] \land \rightrelativ[\foformula_{\att_2}])) \lor
%%     (\leftrelativ[\foformula_{\att_1}][0] \land \foformula_{\att_2})$; 
%%     \item $\foformula_{\ANDop(\att_1, \att_2)}=(\exists x
%%     (\leftrelativ[\foformula_{\att_1}] \land \foformula_{\att_2} \lor \leftrelativ[\foformula_{\att_2}] \land \foformula_{\att_1} \lor
%%     (\leftrelativ[\foformula_{\att_1}][0] \land \foformula_{\att_2}) \lor
%%     (\leftrelativ[\foformula_{\att_2}][0] \land \foformula_{\att_1})$.
%% \end{itemize}
%% We have that  $\foformula_{\ORop(\att_1, \att_2)}$, $\foformula_{\SANDop(\att_1, \att_2)}$ and $\foformula_{\ANDop(\att_1, \att_2)}$ are obtained from conjunctions or disjunctions of leaves formulas and by adding only existential quantifiers to leaves formulas which entails $\foformula_\att \in \foalt[2]$ (\Cref{easyfoalt}(b)). 

On the contrary, proving that $\ADTk[0] \subseteq
\foaltdual[2]$ requires some work. First remark that $\foformula_\epsilonadt \in
\foaltdual[1] \subseteq\foaltdual[2]$. Moreover, recall that for a Boolean formula
$\formula$, 
\[\foformula_\formula \equiv \forall y \exists x ( \lnot (x<y) \land \hat{\formula}(x))\]
where $\formulaFO(x)\egdef \bigvee_{\val \models \formula} \val(x)$.

Indeed, both formulae hold for a trace $\val_1...\val_n$ if, and only if, $\val_n\models\formula$.

%To show $\ADTk[0] \subseteq \foaltdual[2]$, first remark that, for a boolean formula $\formula$, formula $\foformula_\formula=\exists x \forall y (\lnot (x<y) \land \hat{\formula}(x))$ is equivalent to the $\foaltdual[2]$-formula $\forall y \exists x ( \lnot (x<y) \land \hat{\formula}(x))$. Indeed, in both cases, the formula holds for word $\val_1...\val_n$ if, and only if, $\val_n\models\formula$. \tbcom{Revoir la forme, en fonction de l'endroit où est donnée la première fois la formule $\Sigma_2$ pour $\foformula_\formula$}

Recall (\Cref{adt0shape}) that an arbitrary \adt $\att \in \ADTk[0]$,
can be equivalently written $\ORop(\att_1, ..., \att_n)$ where $\att_1,
...\att_n$ are either the empty leaf $\emptyleaf$ or of the form
$\SANDop(\formula_1, ..., \formula_m)$. Consider the following
$\foaltdual[2]$-formula for $\SANDop(\formula_1, ..., \formula_m)$:

\begin{align*}
\theta_{\SANDop(\formula_1, \formula_2, ..., \formula_m)} \egdef  
\forall y \ \exists x_1 \ \exists x_2 \ \ldots \ \exists x_m \quad  & x_1 < x_2 <... < x_m \ \land 
\ y \leq x_m \ \land \\ & \hat{\formula_1}(x_1)\land \hat{\formula_2}(x_2)\land ... \land \hat{\formula_m}(x_m)
\end{align*}

Clearly, for a word $\word \in \Traces$, we have $\word \in
\semword{\SANDop(\formula_1, ..., \formula_m)}$ if, and only if,
$\word \in \Traces \val_1 \concat ... \concat \Traces \val_m$ with
$\val_1 \models \formula_1, ..., \val_m \models \formula_m$, which is
exactly described by $\theta_{\SANDop(\formula_1, \formula_2, ...,
  \formula_m)} \in \foaltdual[2]$.  We now get back to the
construction of the $\foaltdual[2]$-formula for $\att$, that is a
disjunction of either
$\foformula_\emptyleaf$, or formulae of the form
$\theta_{\SANDop(\formula_1, \formula_2, ..., \formula_m)}$). Since they all belong to 
$\foaltdual[2]$, so does their disjunction. 

This concludes the proof of \Cref{lem:adt0tofoalt2foalt2dual}.\\

For \Cref{lem:itemadtktofok+1} of \Cref{lem:adtktofok+1},
   %\tbcom{Juste notation ou une intuition derrière cette découpe? Si oui, la donner!}
 we start with some notations: let $k>0$, for each $d \in \setn$, we define $\ADTk(d) \subseteq \ADTk$ as follows.
 \[
 \left\{
 \begin{array}{l}
   \adtc{k}(0)\egdef \ADTk[k-1] \union \{\Cop(\att_1, \att_2) \,:\,\att_1, \att_2 \in \adtc{k-1}\}\\
\adtc{k}(d+1)\egdef \adtc{k}(d) \union \{\OP(\att_1, \att_2)\,:\, \OP \text{ arbitrary, and }
\att_1, \att_2 \in \adtc{k}(d)\}
 \end{array}
 \right.
 \]

We clearly have 
%% Since for each $k>0$, each leaf \adt is in
%% $\adtc{k}(0)$, each $\att \in \adtc{k}$ of hight $d$
%% that $\att \in \adtc{k}(d)$. 
%
$\underset{d \in \mathbb{N}}{\bigcup} \adtc{k}(d)=\adtc{k}$.

We establish by a double induction over $k>0$ and $d\geq 0$ that for each $d \in \mathbb{N}$, $\ADTk(d) \subseteq
\foalt[k+1]$, which clearly entails $\adtc{k} \subseteq \foalt[k+1]$. 

\begin{description}
\item[$k=1$:] that is to show $\ADTk[1](d) \subseteq \foalt[2]$, for every $d \in \setn$.
\begin{description}
\item[$d=0$:] Let \adt $\att\in \ADTk[1](0)$. If $\att \in \ADTk[0]$, it is immediate by 
  \Cref{lem:adt0tofoalt2foalt2dual} that $\semword{\att}$ is definable by a $\foalt[2]$ formula.

  Otherwise $\att=\Cop(\att_1, \att_2)$ where $\att_1$ and $\att_2\in
  \ADTk[0]$. Again, by \Cref{lem:adt0tofoalt2foalt2dual} there exists
  formulas $\foformula_1 \in \foalt[2]$ and $\foformula_2 \in
  \foaltdual[2]$ that characterize $\semword{\att_1}$ and
  $\semword{\att_2}$ respectively, so that, by the very definition of
  operator $\Cop$, formula $\foformula_1 \land \lnot \foformula_2$
  characterizes $\semword{\att}$ and belongs to $\foalt[2]$ (see
  \Cref{easyfoalt}).

\item[$d>0$:] Let $\att\in \ADTk[1](d)$; note that $\counterdepth(\att)\leq 1$. If $\att \in \ADTk[1](d-1)$,
  we use the induction hypothesis over $d-1$ and we are
  done. Otherwise, $\att=\OP(\att_1,\att_2)$.

  Suppose $\OP \in \{\ORop, \SANDop,\ANDop\}$, and because
  $\att_1,\att_2 \in \ADTk[1](d-1)$, we know by induction over $d-1$
  that there exist $\foformula_1,\foformula_2 \in \foalt[2]$ that
  characterize $\semword{\att_1}$ and $\semword{\att_2}$ respectively.
  For $\OP=\ORop$ (resp.\ $=\SANDop$, $\ANDop$), formula
  $\foformula=\foformula_1 \lor\foformula_2$ (resp.\ $=\exists x
  [\leftrelativ[\foformula_1] \land \rightrelativ[\foformula_2]] \lor
  (\leftrelativ[\foformula_1][0] \land \foformula_2)$, $=\exists x
  [(\leftrelativ[\foformula_1] \land \foformula_{2}) \lor
    (\leftrelativ[\foformula_{2}] \land \foformula_{1})] \lor
  (\leftrelativ[\foformula_{1}][0] \land \foformula_{2}) \lor
  (\leftrelativ[\foformula_{2}][0] \land \foformula_{1})$) 
  characterises $\semword{\OP(\att_1,\att_2)}$ and belongs to $\foalt[2]$
  (see \Cref{foaltrelativization}).

  Now, suppose $\OP=\Cop$. Therefore, $\att_2 \in \ADTk[0](d-1)$
  otherwise $\counterdepth(\att)=2$. By
  \Cref{lem:adt0tofoalt2foalt2dual}, we can assume that $\foformula_2
  \in \foaltdual[2]$ Now formula $\foformula = \foformula_1 \land
  \lnot \foformula_2$ characterises $\semword{\att}$ and is indeed in
  $\foalt[2]$ (see \Cref{easyfoalt}).

\end{description}
\item[$k>1$:] that is to show $\ADTk(d) \subseteq \foalt[k+1]$, for every $d \in \setn$.
  %We show that we show $\att$ is $\foalt[k+1]$-definable for each $\att\in\ADTk(d+1)$.
  %Once again, we proceed by induction over $d$.
 \begin{description}
\item[$d=0$:] Let $\att \in \ADTk(0)$. If $\att \in \ADTk[k-1]$, we
  resort to the induction hypothesis over $k$ to conclude since
  $\foalt[k] \subseteq \foalt[k+1]$. Otherwise, $\att=\Cop(\att_1,
  \att_2)$ with $\att_1$ and $\att_2 \in \ADTk[k-1]$. By induction
  over $k-1$, $\semword{\att_1}$ and $\semword{\att_2}$ can be
  equivalently characterized by $\foalt[k]$-formulas, say
  $\foformula_1$ and $\foformula_2$ respectively. Now the formula $\foformula
  \egdef\foformula_1\land \lnot \foformula_2$ characterizes
  $\semword{\att}$ and, by \Cref{easyfoalt} is clearly in
  $\foalt[k+1]$.

\item [$d>0$:] Let $\att \in\ADTk(d)$. If $\att \in
  \ADTk(d-1)$, we use the induction hypothesis on $d-1$.  Otherwise,
  $\att=\OP(\att_1,\att_2)$ with $\att_1,\att_2 \in\ADTk(d-1)$. We can proceed in a way similar to what we did
  for the previous case $k=1,d>0$, by noticing that for the case where
  $\OP=\Cop$, namely $\att=\Cop(\att_1,\att_2)$, the tree $\att_2 \in
  \ADTk[k-1]$ otherwise $\att$ would not belong to $\ADTk$.

%%   or is of the
%%   form $\OP(\att_1, \att_2)$, with $\att_1, \att_2 \in
%%   \ADTk(d-1)$. The former case is immediate by induction
%%   hypothesis. For the latter case: if $\OP$ is among $\ORop, \SANDop$
%%   or $\ANDop$, then $\foformula_\att$ is constructed on the basis of
%%   conjunction, disjunction and/or by adding existential quantifiers to
%%   the $\foalt[k+1]$-formulas for $\att_1$ and $\att_2$, which entail
%%   $\foformula_\att \in \foalt[k+1]$ (\Cref{easyfoalt}(b)). If $\OP$ is
%%   $\Cop$, $\att_2 \in \ADTk[k-1]$ (otherwise, $\att$ would not belong
%%   to $\ADTk$), and by induction hypothesis over $k-1$,
%%   $\foformula_{\att_2} \in \foalt[k]$. Therefore $\lnot
%%   \foformula_{\att_2} \in \foalt[k+1]$ (by \Cref{easyfoalt}(c)) and
%%   $\foformula_\att$ defined as $\foformula_{\att_1} \land \lnot
%%   \foformula_{\att_2}$ belongs to $\foalt[k+1]$ (by
%%   \Cref{easyfoalt}(b)), Which concludes.
\end{description}

\end{description}

\end{proof}

% \section{Proof of \Cref{lem:fo1inADT0}}
% %----------------------------------
% \label{app:lem:fo1inADT0}

 % \section{Proof of \Cref{Ltoadt}}
 % %----------------------------------
 % \label{app:Ltoadt}

\lemmadotdepth*

\begin{proof}

For \Cref{lem:withdotsdepth}, the proof is conducted by induction over $ell$. We start by showing that $\dotdepth[0]\subseteq\ADTk[2]$. We recall that a language $\lang$ is in $\dotdepth[0]$ if it is finite or co-finite. We distinguish the two cases. 

If $\lang$ is finite, we have that $L=\{\trace_1, ..., \trace_n\}=\{\trace_1\}\union ... \union \{\trace_n\}$ where $\trace_1, ..., \trace_n \in \Traces$.  For a trace $\trace=\val_1, ..., \val_n$, we use the notation $\att_\trace$ to express the \adt $\SANDop(\strict{\val_1}, ..., \strict{\val_n})$. Clearly, $\semword{\att_\trace}=\{\trace\}$ and $\att_\trace\in \ADTk[1].$ Then if we consider the \adt $\att_\lang=\ORop(\att_{\trace_1}, ...,\att_{\trace_n})$,  we have by construction that $\semword{\att_\lang}=\lang$. Moreover $\att_\lang\in \ADTk[1]$ since an \adt $\att_\trace$ only uses non-nested \Cop.

If $\lang$ is co-finite, then it can be written as $\alphabettraces\setminus \lang'$ with $\lang'$ a finite language. Thus, by definition of $\Cop$, if we consider $\att_\lang=\Cop(\true, \att_{\lang'})$, we have $\semword{\att_\lang}=\lang$. We also have $\att_\lang\in \ADTk[2]$. Since $\att_{\lang'}\in \ADTk[1]$. We concludes that $\dotdepth[0]\subseteq\ADTk[2]$. 

Now, for $k>0$, we assume $\dotdepth[\ell-1]\subseteq\ADTk[2\ell]$ and we show that $\dotdepth\subseteq\ADTk[2\ell+2]$.

Let $\lang \in \dotdepth$, so  $\lang$ is a Boolean combination of languages of the form $\lang_1\concat...\concat \lang_n$ with $\lang_1, ..., \lang_n \in \dotdepth[\ell-1]$. With no loss of generality, we can assume the Boolean combination written in disjunctive normal form where each conjunct is of the form $M_1 \inter ... \inter M_m$, where the sets $M_i$ can be written either as $\lang_1\concat...\concat \lang_n$ or as $(\lang_1\concat...\concat \lang_n)^c$ with $\lang_1, ..., \lang_n \in \dotdepth[\ell-1]$. Moreover, We can rewrite each conjunct by using the equality $M_1 \inter ... \inter M_m=(M_1^c\union ... \union M_m^c)^c$ to suppress all intersections of the Boolean expression, it results an expression only using union operators and complement operators. Furthermore, the maximal tower of complementary operator is $2$. By induction hypothesis, we know that an expression $\lang_1\concat...\concat \lang_n$ with $\lang_1, ..., \lang_n \in \dotdepth[\ell-1]$ is definable in $\ADTk[2\ell]$, indeed we consider the $\SANDop$ of the \adts associated with the languages $\lang_1$, ... $\lang_n$. Then, we can use the Boolean expression to construct the adt for $\lang$, where the union is replaced by the $\ORop$ and the complementation is replaced by $\coatt[.]$. Since the maximal tower of complementary operator of the Boolean expression is $2$, we need at most two more nested \Cop to express the full Boolean expression. Therefore, the final \adt is in $\ADTk[2\ell+2]$, which concludes.

\Cref{lem:withdotsdepth} of \Cref{lem:dotdepth} is an immediate corollary of
the fact that $\foalt \subseteq \dotdepth$.

%\lemmafouninadtzero*

%\begin{proof}
Regarding \Cref{lem:fo1inADT0} of \Cref{lem:dotdepth}, let $\foformula
\in \foalt[1]$ a first-order formula of the form $\exists x_1, ...,
\exists x_n \foformula'(x_1, ..., x_n)$ with $\foformula'$
quantifier-free. With no loss of generality, we consider
$\foformula'=C_1\lor ... \lor C_m$ in disjunctive normal form, in
other words $C_1, ..., C_m$ only use conjunctions of literals of the
form $x_i<x_j$, or $\val(x)$, or their negations. As expected, the
\adt corresponding to $\foformula$ is of the form
$\ORop(\att_{C_1},\ldots,\att_{C_m})$, and we now explain how to build
$\att_C$.

If clause $C$ is not satisfiable, we
associate \adt $\false$, otherwise we proceed as follows.

First, we show that with no loss of generality clause $C$ of
$\foformula'$ specifies a linear order over $x_1, ..., x_n$ together
with literals based on predicates $\val(x_i)$.  To do so, we decompose
the clause $C$, according to $C^< \land C^{\val}$, where the first
conjunct gathers all $<$-based literals and the second gathers the
rest.  Now, the sub-clause $C^<$ naturally induces a partial order
between all free variables $x_1, ..., x_n$ of $\foformula'$. It is
easy to see that clause $C^<$ is equivalent to a disjunctive formula
$\foformula_{C^<}$ where each disjunct describes a possible
linearization of this partial order. For example, if
$C^<=x<y \land \lnot y<z$ then $\foformula_{C^<}=(x<y \land y=z) \lor
(x<z \land z<y) \lor (z=x \land x<y) \lor (z<x \land x<y)$, so that
$C$ is equivalent to $(x<y \land y=z \land C^{\val}) \lor (x<z \land
z<y\land C^{\val}) \lor (z=x \land x<y\land C^{\val}) \lor (z<x \land
x<y\land C^{\val})$, which concludes.

In clause $C$ (which now specifies a linear order of the variables), if $x_i=x_j$ with $i\neq j$, we
replace all occurrences of $x_j$ by $x_i$ in $C$ and obtain a clause equivalent to $C$ of the
form $x_{i_1}<...<x_{i_\ell} \land P_1(x_{i_1})\land ...\land P_\ell(x_{i_\ell})$ where each 
$P_k(x_{i_k})$ is the conjunction of literals of the form $\val(x_{i_k})$ and
$\lnot\val(x_{i_k})$. For convenience, we still write $C$ this clause.

We associate with clause $C$ the \adt
$\att_{C}=\attall[\SANDop(\formula_1, ..., \formula_\ell)]$ where
$\formula_i$ is made of the conjunction of the propositional formulas
stemming from the valuations, or negation of valuations, that occur in
$P_k(x_{i_k})$. It can be shown that a trace
$\trace \in \semword{\att_C}$ if, and only, if $\trace\models\exists
x_{i_1} ... \exists x_{i_\ell} x_{i_1}<...<x_{i_\ell} \land P_1(x_{i_1})\land ...\land
P_\ell(x_{i_\ell})$.
%Indeed the $\SANDop$ operator guarantee that the linear order is respected and each leaf guarantee the same property as each $P_k$. 
%\end{proof}

\end{proof}

\section{Complements of \Cref{sec-ADThierarchy}}
\label{app:Ltoadt}

We start this section with some results and definitions to help us to prove \Cref{Ltoadt}. We then conduct the proof of \Cref{Ltoadt}, and we finish with the proof of \Cref{lem-collapse}, a useful result to prove \Cref{expressiveness}.

\begin{lemma}
\label{wolfgangmesure}
    Let $\trace_1, \trace_2,\trace \in \{a,b\}^*$. The following assertions hold:
    \begin{enumerate}
        \item\label{wolfgangmesureitema} $\wolfgangmesure{\trace_1 \concat \trace_2}=\wolfgangmesure{\trace_1}+\wolfgangmesure{\trace_2}$
        \item\label{wolfgangmesureitemb} If $\wolfgangmesure{\trace}=k\geq 0$, then, for each $i \in \{0, 1, ... , k\}$ there exists $\trace'\prefix \trace$ such that $\wolfgangmesure{\trace'}=i$. 
        \item\label{wolfgangmesureitemc} If $\wolfgangmesure{\trace}=k>0$, then there exists $\trace_1$, $\trace_2$ such that $\trace=\trace_1\concat a \concat\trace_2$ with $\wolfgangmesure{\trace_2}=0$ and for each $\trace_2'\prefix\trace_2$, we have $\wolfgangmesure{\trace_2'}\geq0$
    \end{enumerate}
\end{lemma}

\Cref{wolfgangmesureitema} is trivial from the definition of $\wolfgangmesure{.}$, \Cref{wolfgangmesureitemb} stems from the classic "Intermediate Value Theorem" in mathematics -- that is applicable since  extending a trace $\trace$ with a letter only increments or decrements $\wolfgangmesure{\trace}$ by $1$. \Cref{wolfgangmesureitemc} is obtained by an application of \Cref{wolfgangmesureitemb}.

We now start the proof of \Cref{Ltoadt}.

\lemmaltoadt*
\begin{proof}

Recall $\witness \notin \ADTk[k-2]$ (page \pageref{Ltoadt}). We here focus on proving $\witness \notin \ADTk[k+1]$.

We introduce companion languages of $\witness$, that were introduced in \cite{thomas1984application} namely $\witness^+$ and $\witness^-$:

    \begin{itemize}
        %\item $\witness\egdef\{\word \in \Traces $ s.t. $\wolfgangmesure{\word}=0,$ for each $\word'\prefix\word$, $0 \leq \wolfgangmesure{\word'} \leq k$ and there exists $\word''\prefix\word$ such that $\wolfgangmesure{\word''}=k\}$,
        \item $\witness^+\egdef\{\word \in \Traces$ s.t. $ \wolfgangmesure{\word}=k$ and for each $\word'\prefix\word$,  $0 \leq \wolfgangmesure{\word'} \leq k\}$,
        \item $\witness^-\egdef\{\word \in \Traces$ s.t. $ \wolfgangmesure{\word}=-k$ and for each $\word'\prefix\word$,  $-k \leq \wolfgangmesure{\word'} \leq 0\}$,
    \end{itemize}

In \cite{thomas1984application}, alternative definitions of $\witness$, $\witness^+$ and $\witness^-$ are provided:
\begin{itemize}
    \item $\witness[0]=\witness[0]^+=\witness[0]^-\egdef\{\eword\}$;
    \item $\witness[k+1]\egdef(\witness^+a\Traces\cap \Traces b \witness^-) \setminus (\Traces a \witness^+ a \Traces \cup \Traces b \witness^- b \Traces)$;
    \item $\witness[k+1]^+\egdef(\witness^+a\Traces\cap \Traces a \witness^+) \setminus (\Traces a \witness^+ a \Traces \cup \Traces b \witness^- b \Traces)$;
    \item $\witness[k+1]^-\egdef(\witness^-b\Traces\cap \Traces b \witness^-) \setminus (\Traces a \witness^+ a \Traces \cup \Traces b \witness^- b \Traces)$.
\end{itemize}

In the following, we may use the most convenient characterisation of $\witness$, $\witness^+$ and $\witness^-$. Our proof relies on a third characterisation. We build languages $\newwitness$, $\newwitness^+$ and $\newwitness^-$ where we eliminate the $\inter$ operator occurring in the recursive definitions of $\witness$, $\witness^+$ and $\witness^-$. For $\word\in \{a,b\}^*$ we let its \emph{swap} be the word $\dual{\word}$ obtained by swapping $a$ and $b$ in $\word$. We lift this operator to languages like usual: $\dual{L}=\{\dual{\word}: \word\in L\}$. Note that $\witness^-=\dual{\witness^+}$. 

 We set: 
    \begin{itemize}
    \item $\newwitness[0]=\newwitness[0]^+=\newwitness[0]^-\egdef\{\eword\}$;
    \item $\newwitness[k+1]\egdef(\newwitness^+a \Traces b \newwitness^-) \setminus (\Traces a \newwitness^+ a \Traces \cup \Traces b \newwitness^- b \Traces)$;
    %\item $\newwitnessdual[k+1]=(\newwitness^-b \Traces a \newwitness^+) \setminus (\Traces a \newwitness^+ a \Traces \cup \Traces b \newwitness^- b \Traces)$;
    \item $\newwitness[k+1]^+\egdef((\newwitness^+a \Traces a \newwitness^+)
    \union (\underset{i\leq k}{\bigcup} \newwitness[i]a \newwitness^+))
     \setminus (\Traces a \newwitness^+ a \Traces \cup \Traces b \newwitness^- b \Traces)$;
    \item $\newwitness[k+1]^-\egdef \dual{\newwitness[k+1]^+}$.
    %((\newwitness^-b\Traces b \newwitness^-) \union (\underset{i\leq k}{\bigcup} \newwitnessdual[i]b \newwitness^-)) \setminus (\Traces a \newwitness^+ a \Traces \cup \Traces b \newwitness^- b \Traces)$.
\end{itemize}

\begin{lemma}
\label{lem:newwitness}
    For every $k\geq 0$, we have $\witness=\newwitness$, $\witness^+=\newwitness^+$ and $\witness^-=\newwitness^-$.
    
\end{lemma}

\begin{proof}

The proof is conducted by induction over $k$. By definition, $\witness[0]=\newwitness[0]$, $\witness[0]^+=\newwitness[0]^+$ and $\witness[0]^-=\newwitness[0]^-$. For the rest of the proof, we fix $k>0$.

%and assume $\witness=\newwitness$, $\witness^+=\newwitness^+$ and $\witness^-=\newwitness^-$,  we then show $\witness[k+1]=\newwitness[k+1]$, $\witness[k+1]^+=\newwitness[k+1]^+$ and $\witness[k+1]^-=\newwitness[k+1]^-$.

\begin{itemize}
    \item We start to show $\witness[k+1]=\newwitness[k+1]$. Namely, by replacing $\newwitness$ (respectively $\newwitness^+$, $\newwitness^-$) by $\witness$ (respectively $\witness^+$, $\witness^-$) that: 
    \begin{align*}
        (\witness^+a\Traces\cap \Traces b \witness^-) \setminus (\Traces a \witness^+ a \Traces \cup \Traces b \witness^- b \Traces) \\ =(\witness^+a \Traces b \witness^-)\setminus (\Traces a \witness^+ a \Traces \cup \Traces b \witness^- b \Traces)
    \end{align*}
    It is enough to show that $(\witness^+a\Traces\cap \Traces b \witness^-)=(\witness^+a \Traces b \witness^-)$ for the equality to hold. Because $(\witness^+a\Traces\cap \Traces b \witness^-)\supseteq(\witness^+a \Traces b \witness^-)$ is clear, we focus on showing $(\witness^+a\Traces\cap \Traces b \witness^-)\subseteq(\witness^+a \Traces b \witness^-)$.
Let $\word$ be in $\witness^+a\Traces\cap \Traces b \witness^- $. As $\word \in \witness^+a\Traces$, we can write $\word=\traceb_1 \concat a \concat \traceb_2$ with $\traceb_1\in \witness^+$ and let $n$ be the position of the distinguished $a$ occurrence in $\word$. As we also have $\word \in \Traces b \witness^-$, we can write $\word=\tracec_2\concat b\concat \tracec_1$ with $\tracec_1 \in \witness^-$ and let $m$ be the position of the distinguished $b$ occurrence in $\word$. Clearly $n\neq m$. Moreover, we cannot have $m<n$. Indeed, if $m<n$ we can write $\word=\tracec_2\concat b \concat \word' \concat a \concat \traceb_2$ with $\tracec_2 \concat b \concat \word'=\traceb_1 \in \witness^+$ and $\word'\concat a\concat \traceb_2=\tracec_1 \in \witness^-$. Since $\tracec_2$ is a prefix of $\traceb_1\in \witness^+$, we have  $\wolfgangmesure{\tracec_2} \leq k$, and therefore $ \wolfgangmesure{\tracec_2 \concat b}\leq k-1$ (by \Cref{wolfgangmesureitema} of \Cref{wolfgangmesure}). Moreover, since $\word'$ is a prefix of $\tracec_1\in \witness^-$, we have $\wolfgangmesure{\word'}\leq 0$. We conclude that $\wolfgangmesure{\tracec_2\concat b \concat \word'}\leq k-1$ (by \Cref{wolfgangmesureitema} of \Cref{wolfgangmesure}), which contradicts $\tracec_2\concat b \concat \word' \in \witness^+$. Therefore $n<m$, and $\word \in (\witness^+a \Traces b \witness^-)$, which concludes.

\item We now prove that $\witness[k+1]^+=\newwitness[k+1]^+$. Because the inclusion $\witness[k+1]^+\supseteq\newwitness[k+1]^+$ is clear, we focus on $\witness[k+1]^+\subseteq\newwitness[k+1]^+$. Namely, by replacing $\newwitness$ (respectively $\newwitness^+$, $\newwitness^-$) by $\witness$ (respectively $\witness^+$, $\witness^-$) that:
\begin{align*}
    (\witness^+a\Traces\cap \Traces a \witness^+) \setminus (\Traces a \witness^+ a \Traces \cup \Traces b \witness^- b \Traces) \subseteq \\ (\witness^+a \Traces a \witness^+) \union (\underset{i\leq k}{\bigcup} \witness[i]a \witness^+))\setminus (\Traces a \witness^+ a \Traces \cup \Traces b \witness^- b \Traces)
\end{align*} 

Let $\word\in\witness[k+1]^+$. As, on the one hand, $\word\in
\witness^+a\Traces$, we can write $\word=\traceb_1\concat
a\concat\traceb_2$ with $\traceb_1\in \witness^+$ and let $n$ be the
position of the distinguished $a$ occurrence in $\word$. As, on the
other hand, we also have $\word\in \Traces a \witness^+$, we can write
$\word=\tracec_2\concat a\concat\tracec_1$ with $\tracec_1 \in
\witness^+$ and let $m$ be the position of the distinguished $a$
occurrence in $\word$.

First of all, $n\neq m$: indeed, if $n=m$, we can write $\word=
\traceb_1 \concat a\concat \tracec_1$ with $\traceb_1,\tracec_1\in
\witness^+$. Since $\wolfgangmesure{\traceb_1\concat a}=k+1$, we can
write $\traceb_1\concat a=\tracef_1\concat a\concat \tracef_2$ with
$\tracef_2 \in \witness^+$. Since $\tracec_1 \in \witness^+$, we can
write $\tracec_1=a\concat \tracec_1'$. So that $\word=
\tracef_1\concat a\concat \tracef_2\concat a\concat \tracec_1'\in
(\Traces a \witness^+ a \Traces)$ which entails
$\word\notin\witness[k+1]$, leading to a contradiction.

We distinguish the two remaining cases.

\bigskip

\begin{itemize}
    \item If $n<m$, then we have $\word \in (\witness^+a \Traces a \witness^+)$, which entails $\word\in\newwitness[k+1]^+$.
\bigskip
    \item If $m<n$, we can write $\traceb_1 = \traced_1\concat
      a\concat \traced_2 \in \witness^+$ and $\tracec_1 = \traced_2
      \concat a\concat \traced_3 \in \witness^+$, where
      $\word=\traced_1 \concat a \concat \traced_2 \concat a\concat
      \traced_3$. Notice that $\wolfgangmesure{\traced_1}\ge 0$, as
      $\traceb_1 \in \witness^+$, but we show
      that $\wolfgangmesure{\traced_1}= 0$: if
      $\wolfgangmesure{\traced_1} > 0$, we can rewrite
      $\traced_1=\tracef_1\concat a\concat \tracef_2$ with
      $\wolfgangmesure{\tracef_2}=0$ and for each $\tracef_2'$ prefix
      of $\tracef_2$, we have $\wolfgangmesure{\tracef_2'}\geq 0$ (by
      \Cref{wolfgangmesureitemc} of \Cref{wolfgangmesure}). Moreover,
      $\wolfgangmesure{a\concat \traced_2 \concat a\concat
        \traced_3}=k+1$ (by \Cref{wolfgangmesureitema} of
      \Cref{wolfgangmesure} since $\wolfgangmesure{\traced_2 \concat a\concat
        \traced_3}=k$), thus we can write
      $a\concat\traced_2\concat a\concat \traced_3= \traceg_1\concat
      a\concat \traceg_2$ with $\traceg_1\in \witness^+$ (by
      \Cref{wolfgangmesureitemb} of \Cref{wolfgangmesure}). Therefore,
      $\word=\tracef_1\concat a\concat \tracef_2\concat
      \traceg_1\concat a \concat \traceg_2$ with
      $\tracef_2\concat\traceg_1\in \witness^+$(by
      \Cref{wolfgangmesureitema} of \Cref{wolfgangmesure}), so that $\word
      \in (\Traces a \witness^+ a \Traces)$, which contradicts 
      $\word\in\witness[k+1]$.

      Since $\wolfgangmesure{\traced_1}= 0$, we have $\traced_1 \in
      \underset{i\leq k}{\bigcup} \witness[i]$, and we conclude
      $\word=\traced_1\concat a\concat \tracec_1 \in \underset{i\leq
        k}{\bigcup} \witness[i]a \witness^+$. Also because $\word \in
      \witness[k+1]^+$, $\word \notin (\Traces a \witness^+ a \Traces \cup \Traces b \witness^- b \Traces)$ and we obtain $\word \in
      \newwitness[k+1]^+$, which concludes.

\end{itemize}

\bigskip

\item Finally, $\witness[k+1]^-=\newwitness[k+1]^-$ because $\witness[k+1]^-=\dual{\witness[k+1]^+}$ and $\newwitness[k+1]^-=\dual{\newwitness[k+1]^+}$ and we have already shown $\witness[k+1]^+=\newwitness[k+1]^+$.
\end{itemize}
 This conclude the proof of \Cref{lem:newwitness}.
\end{proof}

We use \Cref{lem:newwitness} to construct three \adts $\attwitness, \attwitness^+$ and $\attwitness^-$ such that $\semword{\attwitness}=\witness$, $\semword{\attwitness^+}=\witness^+$ and $\semword{\attwitness^-}=\witness^-$, $\counterdepth(\attwitness)=\counterdepth(\attwitness^+)=\counterdepth(\attwitness^-)=k+1$ which achieves the proof of $\Cref{Ltoadt}$. The proof is conducted by induction over $k\geq 1$.

To capture $\witness[1]$, we propose the following \adt depicted in \cref{fig:attwitness1}:
\[
\attwitness[1] \egdef \Cop(b,\ORop(\attall[\strict{b}], \allattall[\Cop(\adtge[2],\ANDop(\attall[a], \attall[b]))])),\]

We prove $\semword{\attwitness[1]}=(ab)^+=\witness[1]$:
we point to \Cref{ex:adtsexamples} for the semantics for
$\etrue, \coatt$ and $\allattall$.
First of all, $\ANDop(\attall[a], \attall[b])$ defines
the set of all traces with at least one occurrence of $a$ and one occurrence of
$b$. Thus 
%\[
$\semword{\Cop(\adtge[2],\ANDop(\attall[a], \attall[b]))}=aa^+ \union
bb^+$. Write $\att\egdef \Cop(\adtge[2],\ANDop(\attall[a], \attall[b]))$. 
%\]
So, %We can then compute $\allattall[\att]$ 
%\[
$\semword{\allattall[\att]}=\Traces . (aa^+ \union bb^+).\Traces=\Traces . aa.\Traces \union \Traces.bb.\Traces.$ 
%\]
On this basis, \\
$
\semword{\attwitness[1]}=\semword{\Cop(b,\ORop(\attall[\strict{b}], \allattall[\att]))}=\Traces b \setminus (b\Traces\union \Traces . aa.\Traces \union \Traces.bb.\Traces)=(ab)^+=\witness[1]$. 

We now compute $\counterdepth(\attwitness[1])$ since $\counterdepth(\att)=1$, we have $\counterdepth(\attwitness[1])=\counterdepth(\Cop(b,\ORop(\attall[\strict{b}], \allattall[\att])))=\max\{0, \max\{ 1, 1\}+1\}=2$.

Similarly, we define:
\begin{itemize}
\item $\attwitness[1]^+ \egdef \Cop(a,\ORop(\attall[\strict{b}], \allattall[\Cop(\adtge[2],\ANDop(\attall[a], \attall[b]))]))$;
\item $\attwitness[1]^- \egdef \Cop(b,\ORop(\attall[\strict{a}], \allattall[\Cop(\adtge[2],\ANDop(\attall[a], \attall[b]))]))$.
\end{itemize}
As done for $\attwitness[1]$, one can verify that
$\semword{\attwitness[1]^+}=(ab)^*a=\witness[1]^+$ and 
$\semword{\attwitness[1]^-}=b(ab)^*=\witness[1]^-$
and that
$\cdepth(\attwitness[1]^+)=\cdepth(\attwitness[1]^-)=2$.

We make use of $\newwitness$,
$\newwitness^+$ and $\newwitness^-$ to inductively define $\attwitness$, $\attwitness^+$ and
$\attwitness^-$. For readability, we introduce for $k\geq 2$ the subtrees $\attwitnesscount$ and $\attwitnessunion$ and we extend the definition of $\dual{.}$ to \adts by applying the swap to leaves. 

%the following way: $\dual{a}=b, \dual{b}=a$ and $\dual{\OP(\att_1, \att_2)}=\OP(\dual{\att_1}, \dual{\att_2})$ for $\OP\in \{\ORop, \SANDop, \ANDop, \Cop\}$.:

\[\attwitnesscount\egdef\ORop(\SANDop(a, \attwitness[k]^+, \attall[\strict{a}]), \SANDop(b, \attwitness[k]^-, \attall[\strict{b}]))\]

\[ \attwitnessunion\egdef \ORop(\SANDop(\attwitness[1], \strict{a}, \attwitness[k]^+), \SANDop(\attwitness[2], \strict{a}, \attwitness[k]^+),...,\SANDop(\attwitness[k], \strict{a}, \attwitness[k]^+)) \]

And $\attwitness$, $\attwitness^+$ and
$\attwitness^-$ are defined by:
\begin{itemize}
    \item $\attwitness[k]\egdef \Cop(\SANDop(\attwitness[k-1]^+, \strict{a}, \etrue, \strict{b}, \attwitness[k-1]^-), \attwitnesscount[k-1])$
    \item $\attwitness[k]^+\egdef \Cop(\ORop(\SANDop(\attwitness[k-1]^+, \strict{a}, \etrue, \strict{a}, \attwitness[k-1]^+), \attwitnessunion[k-1]), \attwitnesscount[k-1])$
    \item $\attwitness[k]^-\egdef \dual{\attwitness^+}$
\end{itemize}

%\noindent $\attwitness[k]'\egdef \Cop(\SANDop(\attwitness[k-1]^-, \strict{b}, \etrue, \strict{a}, \attwitness[k-1]^+), \ORop(\SANDop(a, \attwitness[k-1]^+, \attall[\strict{a}]), \SANDop(b, \attwitness[k-1]^-, \attall[\strict{b}])))$

The \adts $\attwitness$, $\attwitness^+$ and $\attwitness^-$ are built in such a way that a direct application of \adt semantics yields
 $\witness=\semword{\attwitness}$, $\witness^+=\semword{\attwitness^+}$ and
$\witness^-=\semword{\attwitness^-}$. 

It remains to show that $\counterdepth(\attwitness)=\counterdepth(\attwitness^+)=\counterdepth(\attwitness^-)=k+1$ assuming $\counterdepth(\attwitness[k-1])=\counterdepth(\attwitness[k-1]^+)=\counterdepth(\attwitness[k-1]^-)=k$. 
By definition,  $\counterdepth(\attwitness)=\max\{\counterdepth(\att_1), \counterdepth(\attwitnesscount[k-1])+1\}$ with $\att_1=\SANDop(\attwitness[k-1]^+, \strict{a}, \etrue, \strict{b}, \attwitness[k-1]^-)$. By using \Cref{ex:counterdepth}, we have $\counterdepth(\att_1)=\max\{\counterdepth(\attwitness[k-1]^+),\break 1, 0, 1, \counterdepth(\attwitness[k-1]^-)\}$ and $\counterdepth(\attwitnesscount[k-1])=\max\{0,\counterdepth(\attwitness[k-1]^+), 1, 0, \counterdepth(\attwitness[k-1]^-),1\}$. Applying the induction hypothesis over $\counterdepth(\attwitness[k-1]^+)=\counterdepth(\attwitness[k-1]^-)=k$, we obtain: $\counterdepth(\att_1)=\max\{k, 1, 0, 1, k\}=k$ and $\counterdepth(\att_2)=\max\{0,k, 1, 0,k,1\}=k$, thus $\counterdepth(\attwitness)=\max\{k, k+1\}=k+1$, which concludes. Similarly, we can establish $\counterdepth(\attwitness^+)=\counterdepth(\attwitness^-)=k+1$.

We have shown $\witness=\semword{\attwitness}$ with $\counterdepth(\attwitness)=k+1$ which concludes the proof of \Cref{Ltoadt}.

\end{proof}

%\spcom{the rest is cryptic, we should make a direct translation from star-free {\bf languages} to \adts, instead of regular star-free expressions. So make another job on page \pageref{sophiecomment}.}

% Notice that $\cdepth(\attwitness)=Max\{
% 2+Max \{1, \cdepth(\attwitness[k-1]^+), \cdepth(\attwitness[k-1]^-)\},
% 1+Max\{1, \cdepth(\attwitness[k-1]^+), \cdepth(\attwitness[k-1]^-)\}\}$. Moreover, since
% $\cdepth(\attwitness[0])=\cdepth(\attwitness[0]^+)=\cdepth(\attwitness[0]^-)=2$
% and formulas to compute $\cdepth(\attwitness^+)$ and
% $\cdepth(\attwitness^-)$ are equivalent to the one for
% $\cdepth(\attwitness)$, we can conclude, by induction, that
% $\cdepth(\attwitness)=\cdepth(\attwitness^+)=\cdepth(\attwitness^-)=2+\cdepth(\attwitness[k-1])=2k$.

 % \section{Proof of \Cref{lem-collapse}}
 % %----------------------------------
 % \label{app:lem-collapse}

\begin{restatable}{lemma}{lemmacollapse}
  \label{lem-collapse} If 
 $\ADTk[k_0]=\ADTk[k_0+1]$ for some $k_0>0$, then all $\ADTk$ collapse from $k_0$.
\end{restatable}

\begin{proof}
Let $k>0$ be such that for all $\att \in \ADTk[k+1]$, we have that $\att$ is $\ADTk$-definable. Given
$\att_{k+2} \in \ADTk[k+2]$, we know that for each subtree of the form $\Cop(\att_1,\att_2)$ in  $\att_{k+2}$, we have that $\att_2\in \ADTk[k+1]$. Therefore $\att_2$ is
$\ADTk$-definable. Hence there exists $\att_2'\in \ADTk$ such
that $\att_2\equiv\att_2'$. By replacing $\att_2$ by
$\att_2'$ in $\att_{k+2}$, we do not change its semantics. By applying this procedure over all $\Cop$ operator of minimal depth (ie. $\Cop$ operator having no $\Cop$ in their ancestors), we obtain $\att'_{k+1}\in \ADTk[k+1]$ such that $\att_{k+1}\equiv\att'_{k+2}$. By hypothesis, we know that 
$\att'_{k+1}$ is $\ADTk$-definable. This implies that $\att_{k+1}$ is also
$\ADTk$-definable. We can then extend this result for each $l>k$
by induction.
\end{proof}

% \section{Algorithm for the Membership problem}
% %----------------------------------
% \label{app:algo}

\section{Complements of \Cref{sec-decisionproblems}}
\label{app:sec-decisionproblems}
 
Regarding the upper-bound complexity of \ADTmemb
(\Cref{table:results}), we present here an alternating algorithm using
the following logarithmic space (hence a \PTIME complexity for
\ADTmemb) \Cref{algo:mem}. Remark that our algorithm can be extended
to allow arbitrary \eres (with the Kleene star) as inputs, but this is out of the scope of the paper.
%by adding the two cases $\coatt$ and $\andatt{\att_1}{\att_2}$, with the adequate recursive calls.

\begin{algorithm}[h]
\caption{$\Membership(\att, \trace$)}
\label{algo:mem}
\textbf{Input:}  $\att$ an \adt and $\trace$ a trace \newline
\textbf{Output:} $True$ if $\trace \in \semword{\att}$, $False$ otherwise.
\vspace{-0.4cm}
\begin{multicols}{2}
  \begin{algorithmic}[1]
    \SWITCH{$\att$}
\CASE{$\epsilonadt$}
\RETURN $(\trace=\eword)$
\ENDCASE
\CASE{$\att=\formula$}
\RETURN $(last(\trace)\models \formula)$
\ENDCASE
\CASE{$\att=OP(\att')$}
\RETURN $\Membership(\att', \trace)$
\ENDCASE
\CASE{$\att=\ORop(\att_1, ... ,\att_n)$}
\STATE \EXISTGUESS $\att' \in \{\att_1, ..., \att_n\}$
\RETURN $\Membership(\att', \trace)$
\ENDCASE
\CASE{$\att=\SANDop(\att_1, ..., \att_n)$}
\STATE \EXISTGUESS $i \in \{1, ..., \sizetrace{\trace}\}$
\STATE \FORALLGUESS  $test \in \{first,others\}$
\IF{$test=first$}
\RETURN $\Membership(\att_1, \trace_1...\trace_i)$
\ELSE
\RETURN
\STATE $\Membership(\SANDop(\att_2, ..., \att_n), \trace_{i+1}...\trace_{\sizetrace{\trace}})$
\ENDIF
\ENDCASE
\CASE{$\att=\ANDop(\att_1,..., \att_n)$}
\STATE \EXISTGUESS $(i,a) \in \{1, ..., \sizetrace{\trace}\} \times \{1,..., n\} $
\FORALLGUESS  $test \in \{first, others\}$
\IF{$test=first$}
\RETURN $\Membership(\att_a, \trace_1...\trace_i)$
\ELSE
\RETURN
\STATE $\Membership(\ANDop(\att_1,..., \att_{i-1}, \att_{i+1},..., \att_n), \trace)$
\ENDIF
\ENDCASE
%\STATE
\CASE{$\att=\Cop(\att_1, \att_2)$}
\STATE \FORALLGUESS  $test \in \{first, second\}$
\IF{$test=first$}
\RETURN $\Membership(\att_1, \trace)$
\ELSE
\RETURN $\lnot \Membership(\att_2, \trace)$
\ENDIF
\ENDCASE
\ENDSWITCH
\end{algorithmic}
\end{multicols}
\end{algorithm}
.

\begin{proposition}
 \Cref{algo:mem} solves \ADTmemb in logarithmic space.
\end{proposition}
 
\begin{proof}
We start by showing that \Cref{algo:mem} runs in logarithmic space,
then we prove its correctness.

Since at each recursive call we only need to recall over which factor
of $\trace$ we are computing (in constant space) and over which part
of $\att$ (in constant space) we are pursuing the
computation, \Cref{algo:mem} runs in logarithmic space.

Regarding the correctness of \Cref{algo:mem}, we conduct a proof by induction over $\att$.

The cases where $\att$ is a
leaf (lines 1 to 7) are correct by definition. For the case where
$\att=OP(\att')$ (line 9), then, \Cref{algo:mem} is correct too since
$\semword{\att}=\semword{\att'}$.

If $\att=\ORop(\att_1, ..., \att_n)$, then the case lines 13 to 16 are
correct too since by definition, $\trace \in \semword{\att}$ if
and only if one can find $i \in \{1, ...,n\} $ such that
$\trace \in \semword{att_i}$.

If $\att=\SANDop(\att_1, ..., \att_n)$, then, by associativity,
$\att\equiv \SANDop(\att_1, \SANDop(\att_2, ..., \att_n))$. Moreover,
we have that $\trace=\val_1\ldots\val_m \in \semword{\att}$ if and only if one can find
$i \in \{1, ...,m\} $ such that
$\val_1...\val_i \in \att_1$ and
$\val_{i+1}...\val_m \in \SANDop(\att_2,
..., \att_n)$, which is what is done in lines 18-26.

If $\att=\ANDop(\att_1, ..., \att_n)$, then, by
commutativity and associativity for each $1\leq j \leq n$ we have
$\att\equiv \ANDop(\att_j, \ANDop(\att_1, ..., \att_{j-1}, \att_{j+1},
..., \att_n))$. Moreover, if a trace $\trace\in \ANDop(\att_1,
..., \att_n)$, we can always choose $j$ such that there is $1\leq
k\leq n$, with $k\neq j$ and $\trace \in \semword{\att_k}$. Therefore
$\trace$ is in $\semword{\att}$ if and only if $\trace$ has a prefix
in $\semword{\att_j}$ and $\trace$ is in semantics of $\ANDop(\att_1,
..., \att_{j-1}, \att_{j+1}, ..., \att_n)$. Thus procedure describes
from line 28 to 36 is correct.

Finally, if $\att=\Cop(\att_1, \att_2)$n then trace $\trace \in \semword{\att}$ if, and only if, $\trace \in \semword{\att_1}$ and $\trace \not \in \semword{\att_2}$, which is what is done in lines 38-45. 
\end{proof}

% \section{Proof of Correctness of \Cref{ADTNEklowerbound}}
% %----------------------------------
% \label{app:ADTNEklowerbound}
% %\begin{proof}

\begin{restatable}{proposition}{lemmaadtneklowerbound}
\label{ADTNEklowerbound}
\ADTNE[k] with $k \ge 6$ is not solvable in $\NSPACE(g(k-5,c \sqrt{\frac{n-1}{3}}))$.
\end{restatable}

\begin{proof}    
From \cite[Theorem 4.29]{stockmeyer1974complexity}, for an
\ere $\reg$ of size $n$ and of $\complement*$-depth $d$, there
exists a constant $c$ such that the non-emptiness of
$\complement \reg$ is not solvable in $\NSPACE(g(d-3,
c \sqrt{n}))$. Moreover, from the proof of \Cref{theo:starfreeequalsADT}, it
can be shown that the non-emptiness of $\complement \reg$ reduces to
answering \ADTNE for \adt $\Cop(\true, \att_\reg) \in \ADTk[d+2]$ of
size at most $3n+1$, which concludes.
%% It ensues that, for an attack-defence tree $\att$
%% in $\ADT$ \ADTNE is not solvable in
%% $\NSPACE(g(k-5,c \sqrt{\frac{n-1}{3}})$.
%\end{proof}

\end{proof}

% \section{Proof of Correctness of \Cref{ADTNEkupperbound}}
% %----------------------------------
% \label{app:ADTNEkupperbound}

%\lemmaadtnekupperbound*

\begin{restatable}{proposition}{lemmaadtnekupperbound}
\label{ADTNEkupperbound}
For $k\geq 1$, \ADTNE[k] is in $(k+1)$-\EXPSPACE.
\end{restatable}

\begin{proof}
For an \adt $\att \in \ADTk$, we have a formula
$\foformula_\att \in \foalt[k+1]$ with
$\sizeformula{\foformula_{\att}} \in O(2^{\sizeatt{\att}})$
(\Cref{lem:sizefofromadt}) that is equivalent to $\att$
(see \Cref{lem:adttofo}).  Now, the non-emptiness of $\semword{\att}$
is equivalent to the satisfiability of
$\foformula_\att \in \foalt[k+1]$ which,
by \cite{meyer2006weak} take $k$-exponential time in the size
of $\foformula_{\att}$ and therefore $(k+1)$-exponential time in
$\sizeatt{\att}$.
%
 %%  \spcom{rediscuss the argument/presentation/writing. Make it more
%%   succinct} \spcom{recompute complexity maybe look at $\Pi_\ell$
%%   instead etc. Actually I believe that VALID for $\foalt[\ell]$ is in
%%   $(\ell-1)$-\EXPSPACE and VALID for $\Pi_{\ell}$ is in
%%   $\ell$-EXPSPACE, and I am not sure it SPACE bu rather TIME}.
  
%% For an \adt $\att \in \ADTk$, we have a formula
%% $\foformula_\att \in \foalt[k+2]$ equivalent to $\att$
%% (see \Cref{lem:adttofo}). Now, the non-emptiness of $\semword{\att}$
%% is equivalent to the satisfiability of $\foformula_\att$, which is in
%% turn equivalent to the non-validity of
%% $\lnot \foformula_\att \in \bool{\foalt[k+2]}$ \spcom{actually in $\Pi_{\ell}$}. %\subseteq \foalt[k+3]$.
%% Moreover, it is known that the validity problem for formulas in
%% $\bool{\foalt[\ell]}$ is in $\ell$-\EXPSPACE, which concludes since
%% the size of $\lnot \foformula_\att$ is exponential in $\sizeatt{\att}$
%% (\Cref{lem:sizefofromadt}).
\end{proof}

% \begin{proof}
%  Let's remark that deciding equivalence between the leaf attack-defense tree $\false$ and attack-defense tree $\att$ is equivalent to decide whether $\semword{\att}=\emptyset$ and \ADTNE  for attack-defense trees is non-elementary (see \Cref{ne:nonelementary}).

% Moreover to show that the problem is decidable, we have that, for $\att_1$ and $\att_2$ two attack-defense trees, deciding equivalence between $\att_1$ and $\att_2$ is equivalent to decide whether $\ORop (\Cop(\att_1, \att_2), \Cop(\att_2, \att_1))=\emptyset$ and we know that \ADTNE  is decidable for attack-defense trees (see\Cref{ne:nonelementary}).
% \end{proof}

\end{document}